\documentclass[11pt,a4paper]{article}

\usepackage{amsmath,amsfonts,amssymb,amsthm}
\usepackage{mathbbol}
\usepackage{amsthm}
\usepackage{graphicx,color}
\usepackage{boxedminipage}
\usepackage[linesnumbered, boxed, algosection,noline]{algorithm2e}
\usepackage{framed}
\usepackage{thmtools}
\usepackage{thm-restate}
\usepackage{xspace}
\usepackage{vmargin}
\setmarginsrb{1in}{1in}{1in}{1in}{0pt}{0pt}{0pt}{6mm}
\usepackage{todonotes}
 \usepackage[pdftex, plainpages = false, pdfpagelabels, 
                 bookmarks=false,
                 bookmarksopen = true,
                 bookmarksnumbered = true,
                 breaklinks = true,
                 linktocpage,
                 pagebackref,
                 colorlinks = true,  
                 linkcolor = blue,
                 urlcolor  = blue,
                 citecolor = red,
                 anchorcolor = green,
                 hyperindex = true,
                 hyperfigures
                 ]{hyperref} 
 \usepackage{xifthen}
 \usepackage{tabularx}

 \usetikzlibrary{calc}

 \newcommand{\cost}{{\operatorname{cost}}}

\DeclareMathOperator{\operatorClassNP}{{\sf NP}}
\newcommand{\classNP}{\ensuremath{\operatorClassNP}}

\DeclareMathOperator{\operatorClassFPT}{{\sf FPT}\xspace}
\newcommand{\classFPT}{\ensuremath{\operatorClassFPT}\xspace}
\DeclareMathOperator{\operatorClassW}{{\sf W}}
\newcommand{\classW}[1]{\ensuremath{\operatorClassW[#1]}}
\DeclareMathOperator{\operatorClassParaNP}{{\sf Para-NP}\xspace}
\newcommand{\classParaNP}{\ensuremath{\operatorClassParaNP}\xspace}
\DeclareMathOperator{\operatorClassXP}{{\sf X}P\xspace}
\newcommand{\classXP}{\ensuremath{\operatorClassXP}\xspace}

\newcommand{\wei}{\mathbf{w}}

\newcommand{\Oh}{\mathcal{O}}

\newtheorem{theorem}{Theorem}
\newtheorem{lemma}{Lemma}
\newtheorem{claim}{Claim}[section]
\newtheorem{corollary}{Corollary}
\newtheorem{definition}{Definition}
\newtheorem{observation}{Observation}
\newtheorem{proposition}{Proposition}

\newtheorem{reduction}{Reduction Rule}[section]

\newcommand{\fillin}{{\sf fill\mbox{-}in}}
\newcommand{\sign}{\operatorname{sign}}
 
\newcommand{\intc}{{\sf int\mbox{-}comp}}
\newcommand{\splitc}{{\sf split\mbox{-}comp}}
\newcommand{\cliqc}{{\sf c\mbox{-}comp}}

\newcommand{\pname}{\textsc}
\newcommand{\ProblemFormat}[1]{\pname{#1}}
\newcommand{\ProblemIndex}[1]{\index{problem!\ProblemFormat{#1}}}
\newcommand{\ProblemName}[1]{\ProblemFormat{#1}\ProblemIndex{#1}{}\xspace}

  \newcommand{\probWdcS}{\ProblemName{Weighted $d$-colorable Subgraph}}
    \newcommand{\probWHcS}{\ProblemName{Weighted $H$-colorable Subgraph}}
    \newcommand{\probWddS}{\ProblemName{Weighted $d$-degenerate Subgraph}}
 \newcommand{\probWVC}{\ProblemName{Weighted Vertex Cover}}
\newcommand{\probWIS}{\ProblemName{Weighted Independent Set}}
 \newcommand{\probDS}{\ProblemName{Dominating Set}}

\newcommand{\probWFVS}{\ProblemName{Weighted Feedback Vertex Set}}
\newcommand{\probWCFVS}{\ProblemName{Weighted Connected Feedback Vertex Set}}
\newcommand{\probIS}{\ProblemName{Independent Set}}
\newcommand{\probWOCT}{\ProblemName{Weighted Odd Cycle Transversal}}
\newcommand{\probWBS}{\ProblemName{Weighted Bipartite Subgraph}}
\newcommand{\probmWIF}{\ProblemName{Weighted Induced Forest}}
\newcommand{\probCL}{\ProblemName{Clique}}
\newcommand{\probWCL}{\ProblemName{Weighted Clique}}

\newcommand{\probWICP}{\ProblemName{Weighted Induced Cycle Packing}}
\newcommand{\probWCVC}{\ProblemName{Weighted Connected Vertex Cover}}
\newcommand{\probCOL}{\ProblemName{Coloring}}

\newcommand{\chordalmke}{\text{\sc{Chordal}}-ke}
\newcommand{\intervalmke}{\text{\sc{Interval}}-ke}
\newcommand{\splitmke}{\text{\sc{Split}}-ke}
\newcommand{\compmke}{\text{\sc{Complete}}-ke}

\newcommand{\chordal}{\text{\sc{Chordal}}}

\newcommand{\NP}{{\ensuremath{\rm{NP}}}}
\newcommand{\coNP}{{\ensuremath{\rm{coNP}}}}

 \DeclareMathOperator{\poly}{poly}
\newcommand{\compass}{\coNP\subseteq \NP/\poly}
\newcommand{\ncompass}{\coNP\nsubseteq \NP/\poly}



\newlength{\RoundedBoxWidth}
\newsavebox{\GrayRoundedBox}
\newenvironment{GrayBox}[1]%
   {\setlength{\RoundedBoxWidth}{.93\textwidth}
    \def\boxheading{#1}
    \begin{lrbox}{\GrayRoundedBox}
       \begin{minipage}{\RoundedBoxWidth}}%
   {   \end{minipage}
    \end{lrbox}
    \begin{center}
    \begin{tikzpicture}%
       \node(Text)[draw=black!20,fill=white,rounded corners,%
             inner sep=2ex,text width=\RoundedBoxWidth]%
             {\usebox{\GrayRoundedBox}};
        \coordinate(x) at (current bounding box.north west);
        \node [draw=white,rectangle,inner sep=3pt,anchor=north west,fill=white] 
        at ($(x)+(6pt,.75em)$) {\boxheading};
    \end{tikzpicture}
    \end{center}}     

\newenvironment{defproblemx}[2][]{\noindent\ignorespaces%
                                \FrameSep=6pt%
                                \parindent=0pt%
                \vspace*{-1.5em}
                \ifthenelse{\isempty{#1}}{%
                  \begin{GrayBox}{\textsc{#2}}%
                }{%
                  \begin{GrayBox}{\textsc{#2}  parameterized by~{#1}}%
                }
                \begin{tabular*}{\textwidth}{@{\hspace{.1em}} >{\itshape} p{1.8cm} p{0.8\textwidth} @{}}%
            }{
                \end{tabular*}%
                \end{GrayBox}%
                \ignorespacesafterend
            }

\newcommand{\defproblema}[3]{
  \begin{defproblemx}{#1}
    Input:  & #2 \\
    Task: & #3
  \end{defproblemx}
}%

\usepackage{fdsymbol}

\newenvironment{subproof}[1][\proofname]{%
  \begin{proof}[#1]%
}{%
  \end{proof}%
}

\begin{document}
\title{Subexponential parameterized algorithms and kernelization on almost chordal graphs%
\thanks{The research leading to these results have  been supported by the Research Council of Norway via the project ``MULTIVAL" (grant no. 263317).}
}

\author{
Fedor V. Fomin\thanks{
Department of Informatics, University of Bergen, Norway.} \addtocounter{footnote}{-1}
\and
Petr A. Golovach\footnotemark{} 
}

\date{}

\maketitle

\begin{abstract}
We study algorithmic properties of the graph class  $\chordalmke$, that is,   graphs that can be turned into a chordal graph by adding at most $k$ edges
 or, equivalently, the class of graphs of  fill-in at most $k$.
We discover that a number of fundamental intractable optimization problems being parameterized by $k$ admit \emph{subexponential}   algorithms on graphs from  $\chordalmke$.  While various parameterized algorithms on graphs for many structural parameters like vertex cover or treewidth can be found in the literature, up to the Exponential Time Hypothesis (ETH), the existence of subexponential parameterized algorithms for most of the structural parameters and optimization problems is highly unlikely. This is why we find  the algorithmic behavior of the  ``fill-in parameterization''  very unusual. 

 Being intrigued by this   behaviour,    
we identify a large class of optimization problems  on $\chordalmke$ that admit  algorithms with the typical  running time $2^{\Oh(\sqrt{k}\log k)}\cdot n^{\Oh(1)}$. Examples of the problems from this class are    finding an independent set of maximum weight,    finding a feedback vertex set
or  an odd cycle transversal of minimum weight, or the problem of finding a maximum induced planar subgraph.  On the  other hand, we  show that for some   fundamental optimization problems,  like finding an optimal graph coloring or finding a maximum clique, are FPT on $\chordalmke$ when parameterized by $k$ but do not admit subexponential in $k$ algorithms unless ETH fails. 

Besides subexponential time algorithms, the  class  of  $\chordalmke$ graphs appears to be appealing from the perspective of kernelization (with parameter $k$).
While it is possible to show that most of the weighted variants of optimization problems do not admit polynomial in $k$ kernels on  $\chordalmke$ graphs, this does not exclude the existence of Turing kernelization and kernelization for unweighted graphs.  In particular, we  
construct a polynomial Turing kernel for \probWCL on $\chordalmke$ graphs.
For  (unweighted) \probIS  we design  polynomial kernels on  two interesting subclasses of $\chordalmke$, namely,    $\intervalmke$ and  $\splitmke$ graphs. 

\end{abstract}

\section{Introduction}\label{sec:intro}

Many NP-hard  graph optimization problems are solvable in polynomial or even linear time when the input of the problem is restricted to a special graph class. For example, the chromatic number of a perfect graph can be computed in polynomial time  \cite{GrotLS93}, the \textsc{Feedback Vertex Set} problem is solvable in polynomial time on chordal graphs~\cite{Gavril72}, and \textsc{Hamiltonicity} on interval graphs~\cite{Keil85}. 
From the perspective of parameterized complexity,  the natural question here is how stable are these nice algorithmic properties of graph classes subject to some perturbations. For example, if an input  $n$-vertex graph $G$ is not chordal, but can be turned into a chordal graph by adding at most $k$ edges,  how fast can we solve \textsc{Feedback Vertex Set} on $G$? 
Can we solve the problem in polynomial time for constant $k$? Or maybe for $k=\log{n}$ or even for $k=poly(\log{n})$? 
A word of warning is on order here.
Since  an algorithm for \textsc{Feedback Vertex Set} of running time $2^{o(n)}$ will refute the Exponential Time Hypothesis (ETH)
 of Impagliazzo, Paturi and Zane~\cite{ImpagliazzoP99,ImpagliazzoPZ01}, and because $k\leq \binom{n}{2}$, the existence of an  algorithm of running time $2^{k^{1/2 -\varepsilon}}\cdot n^{\Oh(1)}$ for some $\varepsilon >0$ (which is polynomial for $k=(\log{n})^{2/(1-2\varepsilon)}$) is unlikely. Interestingly, as we shall see,   \textsc{Feedback Vertex Set}  (and many other problems) are solvable in time $2^{k^{1/2}\log{k}}\cdot n^{\Oh(1)}$.

Leizhen Cai in  \cite{Cai03a} introduced a convenient notation for ``perturbed'' graph classes.   Let $\mathcal{F}$ be a  class of graphs, then  $\mathcal{F}-ke$ (respectively $\mathcal{F}-ve$)  is the class of those graphs that can be obtained from a member of $\mathcal{F}$ by deleting at most $k$ edges (respectively vertices). Similarly one can define classes   $\mathcal{F}+ke$ and  $\mathcal{F}+ve$. 
Then for any class $\mathcal{F}$ and optimization problem  $\mathcal{P}$ that can be solved in polynomial time on $\mathcal{F}$, the natural question is whether $\mathcal{P}$ is fixed-parameter tractable parameterized by $k$,  the ``distance'' to  $\mathcal{F}$. 
 
 In this paper we obtain several algorithmic results on  the parameterized complexity of optimization problems on  $\mathcal{F}-ke$, where $\mathcal{F}$ is the class of chordal graphs. Let us remind that  a graph $H$ is {\em chordal} (or {\em triangulated}) if every
cycle of length at least four has a chord, i.e., an edge between
two nonconsecutive vertices of the cycle. We denote by $\chordalmke$ the class of graphs that can be made  chordal graph by adding at most $k$ edges.  While parameterized algorithms for various problems on the class of $\chordalmke$ graphs  were studied (see the section on previous work),  our work introduces the first  subexponential parameterized  algorithms on this graph class.   
 We prove the following.

 \medskip\noindent\emph{Subexponential parameterized algorithms}. We discover a large class of optimization problems on graph class $\chordalmke$ that  are solvable in time $2^{\Oh(\sqrt{k}\log{k})} \cdot n^{\Oh(1)}$. Examples of such optimization problems 
are: the problem of finding an induced $d$-colorable subgraph of maximum weight (which generalizes 
\probWIS  for $d=1$ and  \probWBS for $d=2$); the problem of finding a maximum weight induced subgraph  admitting a homomorphism into a fixed graph $H$; the problem of finding an induced  $d$-degenerate subgraph of maximum weight and its variants like 
\probmWIF (or, equivalently, \probWFVS),
\textsc{Weighted Induced Tree}, \textsc{Induced Planar Graph},
\textsc{Weighted Induced Path (Cycle)} or
 \probWICP;  as well as various connectivity variants of these problems like \probWCVC and  \probWCFVS. This implies that all these problems are solvable in polynomial time for $k=(\frac{\log{n}}{\log\log{n}})^2$. On the other hand,
  we refute (subject to ETH)  existence of a subexponential time   $2^{o(k)}\cdot n^{\Oh(1)}$ algorithms  on graphs in  $\chordalmke$ for \probCOL  and \probCL.   Moreover, our lower bounds hold for way more restrictive graph class
  $\compmke$, the graphs within $k$ edges from a complete graph. We also show that both problems are     fixed-parameter tractable (\classFPT)  (parameterized by $k$) on  $\chordalmke$ graphs.

 \medskip\noindent\emph{Kernelization}. It follows almost directly from the previous work~\cite{JansenB13,BodlaenderJK14} 
that   \probWIS, \probWVC , \probWBS,  \probWOCT, \probWFVS and \textsc{Weighted Clique}
do not admit a polynomial in $k$ kernel   (unless $\ncompass$)  on   $\compmke$ and hence on $\chordalmke$.   Interestingly, these lower bounds do not refute the possibility of polynomial Turing kernelization or kernelization for unweighted variants of the problems. 
Indeed, we show that  \probWCL on $\chordalmke$ parameterized by $k$ admits a Turing kernel.  For unweighted  \probIS  we show that the problem admits polynomial in $k$ kernel on  graph classes   
 $\intervalmke$  and  $\splitmke$ (graphs that can be turned into an interval or split graphs, corrspondingly, by adding at most $k$ edges).
  
 \medskip\noindent\textbf{Previous work.} Chordal graphs form an important subclass of perfect graphs. These graphs were also intensively studied from the algorithmic perspective. We refer to books \cite{BrandLeSpiGraphclasses99,Golumbic80,vandenberghe2015chordal} for introduction to chordal graphs and their algorithmic properties. 
 
 The problem of determining whether a graph $G$ belongs to $\chordalmke$, that is checking whether $G$ can be turned into a chordal graph by adding at most $k$ edges,   is known in the literature as the \textsc{Minimum Fill-in} problem.  
 The name fill-in is due to  the fundamental problem arising in sparse matrix computations which was studied intensively in the past   \cite{Parter61,Rose72}.   The survey of Heggernes    \cite{Heggernes06} gives an overview of techniques and applications of  minimum and minimal triangulations.

\textsc{Minimum Fill-in}  (under the name \textsc{Chordal Graph Completion}) was one of the 12 open  problems presented at the end of the first edition of Garey and Johnson's book  \cite{GareyJ79} and  it was  proved to be NP-complete by Yannakakis 
\cite{Yannakakis81}. Kaplan et al.  proved  that \textsc{Minimum Fill-in}  is  fixed parameter tractable by giving an algorithm of running time  $16^k\cdot n^{\Oh(1)}$ in \cite{focs/KaplanST94}. There was a chain of  algorithmic improvements resulting in decreasing the constant in the base of  the exponents \cite{KaplanST99,Cai96,Bodlaender:2011lr} resulting with a subexponential algorithm of running time  $2^{\Oh(\sqrt{k}\log{k})} \cdot n^{\Oh(1)}$ \cite{FominV13}.
 A significant amount of work in parameterized algorithms is devoted to recognition problems of classes   
$\mathcal{F}- ke$, $\mathcal{F}+ke$, $\mathcal{F}- kv$, and $\mathcal{F}+kv$ for chordal graphs and various subclasses of chordal graphs \cite{AgrawalLMSZ19,AgrawalM0Z19,pic-kernel,BliznetsFPP18,CaoM15,Cao16,Cao17,FominSV13,JansenP18,Marx10,yngve:ic}.

 Parameterized algorithms, mostly for graph coloring problems, were studied on  perturbed chordal graphs and subclasses of this graph class
\cite{Cai03a,takenaga2006vertex}.  Among other results, Cai \cite{Cai03a} proved that \probCOL (the problem of computing the chromatic number of a graph) is  \classFPT  (parameterized by $k$) on $\splitmke$ graphs. Marx~\cite{Marx06aParaCol} proved that \textsc{Coloring} is \classFPT on $\textsc{Chordal}+ke$ and   $\textsc{Interval}+ke$ graphs but is \classW{1}-hard on $\textsc{Chordal}+kv$ and   $\textsc{Interval}+kv$ graphs.
Jansen and Kratsch \cite{jansen2013data}
proved that for  every fixed integer $d$, the problems \textsc{$d$-Coloring} and \textsc{$d$-List Coloring}  admit polynomial kernels on the parameterized graph classes   $\textsc{Split} + kv$,    $\textsc{Cochordal} + kv$, and $\textsc{Cograph} + kv$.

 Liedloff,  Montealegre, and Todinca
\cite{DBLP:journals/algorithmica/LiedloffMT19} gave a general theorem establishing fixed-parameter tractability for a large class of optimization problems. Let $\mathcal{C}_{poly}$  be a class of graphs    having at most  $poly(n)$ minimal separators. (Since every chordal graph has 
at most $n$
 minimal separators, the class of chordal graphs is a subclass of  $\mathcal{C}_{poly}$.)
Let 
$\varphi$ be   a Counting Monadic Second Order Logic  (CMSO) formula,  $G$ be a graph,  and $t\geq 0$ be an integer.   Liedloff,  Montealegre, and Todinca
 proved that on graph class $\mathcal{C}_{poly}+kv$, the following generic problem 
  \begin{equation*}\label{eq:opt_phi} 
\begin{array}{ll}
\mbox{Max}  &  |X|   \\
\mbox{subject to } &  \mbox{There is a set }   F\subseteq V (G)   \mbox{ such that }  X\subseteq F;       \\
 &  \mbox{The treewidth of   }  G[F]    \mbox{ is at most }  t ;     \\
 &  (G[F],X)\models\varphi.  
\end{array}
\end{equation*}
is solvable in time 
 $\Oh( n^{\Oh(t)}\cdot f(t,\varphi,k))$, and thus is  
  fixed-parameter tractable parameterized by $k$. 
   The problem generalizes many classical algorithmic problems like \textsc{Maximum Independent Set},   \textsc{Longest Induced Path}, \textsc{Maximum Induced Forest}, and different packing problems, see 
   \cite{DBLP:journals/siamcomp/FominTV15}.
   
    Since the class  $\mathcal{C}_{poly}+kv$ contains  $\chordalmke$,  the work of Liedloff  et al. 
\cite{DBLP:journals/algorithmica/LiedloffMT19}  yields that all these problems are fixed-parameter tractable on  $\chordalmke$ graphs parameterized by $k+t+|\varphi|$.  However, the theorem of Liedloff et al. cannot be used to derive our results. First, this theorem provides FPT algorithm  only for problems of finding an  induced subgraph of constant treewidth, which   is not the case in our situation. Second, even for graphs of treewidth $0$, their technique does not derive parameterized algorithms with subexponential running times.


 \medskip\noindent\textbf{Organization of the paper.} The remaining part of the paper is organized as follows. In Section~\ref{sec:prelim}, we introduce notation and provide some useful auxiliary results. In Section~\ref{sec:subexpalgo}, we discuss subexponential algorithms on $\chordalmke$. Section~\ref{sec:lower} contains the lower bounds for 
\probCOL and \probCL on $\chordalmke$. Sections~\ref{sec:kern}--\ref{sec:clique} are devoted to kernelization. In Section~\ref{sec:kern}, we give lower bounds and construct a polynomial Turing kernel for \probWCL on $\chordalmke$. In Sections~\ref{sec:interval} and \ref{sec:clique}, we construct polynomial kernels for \probIS on $\intervalmke$ and  $\splitmke$ respectively. We conclude in Section~\ref{sec:concl} by some open problems.

\section{Preliminaries}\label{sec:prelim}

\noindent{\bf{Graphs.}}
All graphs considered in this paper are assumed to be  simple, that is, finite undirected graphs without loops or multiple edges.
For each of the graph problems considered in this paper, we let $n=|V(G)|$ and $m=|E(G)|$ denote the number of vertices and edges,
respectively, of the input graph $G$ if it does not create confusion. For a set $X\subseteq V(G)$, $\binom{X}{2}$ denotes the set of pairs of distinct vertices of $X$. 
For a graph $G$ and a subset $X\subseteq V(G)$ of vertices, we write $G[X]$ to denote the subgraph of $G$ induced by $X$.
We write $G-X$ to denote the subgraph of $G$ induced by $V(G)\setminus X$, and we write $G-u$ instead of $G-\{u\}$ for a single element set. 
Similarly, for an edge set $A$, $G-A$ denotes the graph $G\rq{}=(V(G),E(G)\setminus A)$, and for a set of pairs of vertices $A\subseteq \binom{V(G)}{2}$, $G+A$ is the graph $G\rq{}=(V(G),E(G)\cup A)$. For a single-element set $\{e\}$, we use  $G+e$ for  $G+\{e\}$.  For $A\subseteq \binom{V(G)}{2}$, $G\bigtriangleup A$ denotes the graph $G'=(V(G),E(G)\bigtriangleup A)$.
For a vertex $v$, we denote by $N_G(v)$ the \emph{(open) neighborhood} of $v$, i.e., the set of vertices that are adjacent to $v$ in $G$.
The \emph{closed neighborhood} $N_G[v]$ is $N_G(v)\cup \{v\}$.
For a set of vertices $X\subseteq V(G)$, $N_G[X]=\bigcup_{v\in X}N_G[v]$ and $N_G(X)=N_G[X]\setminus X$.
The \emph{degree} of a vertex $v$ is $d_G(v)=|N_G(v)|$.
The complement of a graph $G$ is the graph $\overline{G}$ with $V(\overline{G})=V(G)$ such that two distinct vertices are adjacent in  $\overline{G}$ if and only if they are not adjacent in $G$.
A \emph{(proper) $\ell$-coloring} of a graph $G$ is an assignment $c\colon V(G)\rightarrow \{1,\ldots,\ell\}$ of \emph{colors} $1,\ldots,\ell$ to the vertices of $G$ in such a way that adjacent vertices get distinct colors; a graph $G$ is \emph{$\ell$-colorable} if it has an $\ell$-coloring. 
A graph is \emph{$d$-degenerate} for a nonnegative integer $d$, if every induced subgraph of $G$ has a vertex of degree at most $d$.  An equivalent way of defining  $d$-degenerate graph is in terms of coloring   orderings~\cite{MR0193025}. A vertex ordering of a graph 
is  a \emph{$d$-coloring} ordering, if each vertex  has at most $d$ neighbors that are after it in the ordering.  Then a graph $G$ is $d$-degenerate if and only if it admits an $d$-coloring ordering. 

For a graph class $\mathcal{C}$ and a nonengative integer $k$, $\mathcal{C}-ke$ denotes the class of all graphs $G$ such that there is a set $A\subseteq \binom{V(G)}{2}\setminus E(G)$ of size at most $k$ such that $G+A\in\mathcal{C}$. In words, this means that $\mathcal{C}-ke$ contains graphs that can be turned to be graphs of $\mathcal{C}$ by at most $k$ edge additions. For a set   $A\subseteq \binom{V(G)}{2}\setminus E(G)$ such that $G+A\in\mathcal{C}$, we say that $A$ is a \emph{$\mathcal{C}$-modulator}. 

\smallskip\noindent\textbf{Graph classes.}
A graph $G$ is \emph{chordal} (or \emph{triangulated}) if it does not contain an induced cycle of length at least four. In other words, 
every cycle of length at least four has a chord, i.e., an edge whose end-vertices are nonconsecutive vertices of the cycle.
The intersection graph of a family of intervals of the real line is called an \emph{interval} graph; it is also said that $G$ is an interval graph if there is a family of intervals (called \emph{interval model} or \emph{representation}) such that $G$ is isomorphic to the intersection graph of this family. 
A graph $G$ is said to be \emph{split} if its vertex set can be partition into independent set and a clique. 
 We refer to~\cite{BrandLeSpiGraphclasses99,Golumbic80} for  detailed introduction to these graph classes. Notice that interval and split graphs are chordal.

A \emph{triangulation} (or a \emph{chordal complementation}) of a graph $G$ is a chordal supergraph $H$ with $V(H)=V(G)$.  The \emph{size} of the triangulation is $|E(H)|-|E(G)|$. 
The \emph{fill-in} of a  graph $G$, denoted $\fillin(G)$, is the minimum integer $k$ such that $G\in \chordalmke$ 
or, in other words, fill-in is the minimum number of edges whose addition makes the graph chordal. 
An \emph{interval complementation} of a graph $G$ is an interval supergraph $H$ with $V(H)=V(G)$.  Similarly, a \emph{split complementation} of $G$ is a split supergraph $H$  and a \emph{clique complementation} is a complete  supergraph with $V(H)=V(G)$. The \emph{size of interval  (split, clique) completion} is $|E(H)|-|E(G)|$ and we denote the minimum size of an interval (split, clique)
 complementation by $\intc(G)$ ($\splitc(G)$, $\cliqc(G)$ respectively). 
Clearly, $G$ has an interval (split, clique) complementation of  size at most $k$ if and only if $G\in\intervalmke$ ($\splitmke$, $\compmke$).
It is easy to see that  $\cliqc(G)=\binom{|V(G)|}{2}-|E(G)|$, and it is well-known that it is NP-hard to compute $\fillin(G)$~\cite{Yannakakis81} and $\intc(G)$~\cite{GareyJ79} and the same holds for $\splitc(G)$~\cite{NatanzonSS01}.

We will make use of the following observation.

\begin{observation}\label{obs:bounds}
For every graph $G$, $\cliqc(G)\geq \intc(G)\geq \fillin(G)$ and  $\cliqc(G)\geq \splitc(G)\geq \fillin(G)$.
\end{observation}

In particular, this observation implies that complexity lower bounds obtained for graph problems parameterized by the clique completion size hold for the same problems when they are parmeterized by the interval or split completion or by the fill-in, and the hardness for the interval or split completion parameterization implies the hardness for the fill-in parameterization.

Natanzon, Shamir and Sharan~\cite{NatanzonSS98} proved that fill-in admits a polyopt approximation. 

\begin{proposition}[\cite{NatanzonSS98}]\label{prop:fillin-appr}
There is a polynomial algorithm that, given a graph $G$ and a nonnegtive integer $k$, either correctly reports that $\fillin(G)>k$ or returns a triangulation of $G$ of size at most $8k^2$.
\end{proposition}

\medskip\noindent{\bf{Tree decompositions.}} A   {\em tree decomposition} of
a graph $G$ is a pair $\mathcal{T}=(T,\{X_t\}_{t\in V(T)})$, where $T$ is a tree whose every node $t$ is assigned a vertex subset $X_t\subseteq V(G)$, called a bag,
such that the following three conditions hold:
\begin{description}
\item[(T1)] $\bigcup_{t\in V(T)} X_t =V(G)$. In other words, every vertex of $G$ is in at least one bag.
\item[(T2)] For every $uv\in E(G)$, there exists a node $t$ of $T$ such that bag $X_t$ contains both $u$ and $v$.
\item[(T3)] For every  $u\in V(G)$, the set $T_u = \{t\in V(T) | u\in X_t\}$, i.e., the set of nodes whose corresponding  bags contain $u$, induces
a connected subtree of $T$.
\end{description}

To distinguish between the vertices of the decomposition tree $T$ and the vertices of
the graph $G$, we will refer to the vertices of $T$ as {\em{nodes}}. 
By the classical result due to Buneman and Gavril~\cite{Buneman74,Gavril74}, 
every chordal graph $G$  has a tree decomposition   
such that each bag of the decomposition is a maximal clique of $G$.
Such a tree decomposition is referred as a \emph{clique tree} of the chordal graph $G$. 
 
  It is more convenient to describe dynamic programming algorithms on tree decompositions of special nice form. 
 We think of nice tree decompositions as rooted trees. That is, for a tree decomposition $(T,\{X_t\}_{t\in V(T)})$ we distinguish one vertex $r$ of $T$ which will be the root of $T$.  We   say that such a rooted tree decomposition $(T,\{X_t\}_{t\in V(T)})$ is {\em{nice}} if the following conditions are satisfied: 
\begin{itemize}
\item $X_r=\emptyset$ and $X_\ell=\emptyset$ for every leaf $\ell$ of $T$. In other words, all the leaves as well as the root contain empty bags. 
\item Every non-leaf node of $T$ is of one of the following three types:
\begin{itemize}
  \item{\bf{Introduce node}}: a node $t$ with exactly one child $t'$
    such that $X_t=X_{t'}\cup\{v\}$ for some vertex $v\notin X_{t'}$; we say
    that $v$ is {\emph{introduced}} at $t$.
  \item{\bf{Forget node}}: a node $t$ with exactly one child $t'$ such
    that $X_t=X_{t'}\setminus \{w\}$ for some vertex $w\in X_{t'}$; we say
    that $w$ is {\emph{forgotten}} at $t$.
  \item{\bf{Join node}}: a node $t$ with two children $t_1,t_2$ such
    that $X_t=X_{t_1}=X_{t_2}$.
\end{itemize}
\end{itemize}
Throughout the paper, given a nice tree decomposition $(T,\{X_t\}_{t\in V(T)})$ of a graph $G$, we denote by $V_t$ the union of the bags in the subtree of $T$ rooted in a node $t\in V(T)$.
 
We will be using the following proposition, see e.g. \cite{Marx06aParaCol}.

\begin{proposition}\label{lem02:nice-tw}
Every $n$-vertex chordal graph $G$ admits a nice tree decomposition \linebreak $\mathcal{T}=(T,\{X_t\}_{t\in V(T)})$ such that 
every bag $X_t$ of $G$ is a clique.
This decomposition 
 has at most $\Oh(n^2)$ nodes and can be constructed in polynomial  time.  
\end{proposition}

\noindent
{\bf Parameterized Complexity and Kernelization.}
We refer to the  books~\cite{CyganFKLMPPS15,DowneyF13,FominLSM19} for the detailed introduction to the field. Here we only
briefly review the basic notions.

Parameterized Complexity is a two dimensional framework
for studying the computational complexity of a problem. One dimension is the input size~$n$ and the other is a \emph{parameter}~$k$ associated with the input.

A parameterized problem is said to be \emph{fixed parameter tractable} (or \classFPT) if it can be solved in time $f(k)\cdot n^{\Oh(1)}$ for some function~$f$. 

Parameterized complexity theory also provides tools for obtaining complexity lower bounds. Here we use lower bounds based on \emph{Exponential Time Hypothesis (ETH)} formulated by Impagliazzo,  Paturi and Zane~\cite{ImpagliazzoP99,ImpagliazzoPZ01}.  For an integer $k\geq 3$, let $q_k$ be the infimum of the real numbers $c$ such that the \textsc{$k$-Satisfiability} problem can be solved in time $\Oh(2^{c n})$, where $n$ is the number of variables. Exponential Time Hypothesis states that $\delta_3>3$. In particular, ETH implies that   \textsc{$k$-Satisfiability} cannot be solved in time $2^{o(n)}n^{\Oh(1)}$.

A \emph{compression} of a parameterized problem $\Pi_1$ into a (non-parameterized) problem $\Pi_2$
is a polynomial algorithm that maps each instance $(I,k)$ of $\Pi_1$ with the input~$I$ and the parameter~$k$ to an instance $I'$ of $\Pi_2$ such that
\begin{itemize}
\item[(i)] $(I,k)$ is a yes-instance of $\Pi_1$ if and only if $I'$ is a yes-instance of $\Pi_2$, and
\item[(ii)] $|I'|$ is bounded by~$f(k)$ for a computable function~$f$.
\end{itemize}
The output $I'$ is also called a \emph{compression}. The function~$f$ is said to be the \emph{size} of the compression. A compression is \emph{polynomial} if~$f$ is polynomial.
A \emph{kernelization } algorithm for a parameterized problem $\Pi$ is a polynomial algorithm that maps each instance $(I,k)$ of $\Pi$ to an instance $(I',k')$ of $\Pi$ such that
\begin{itemize}
\item[(i)] $(I,k)$ is a yes-instance of $\Pi$ if and only if $(I',k')$ is a yes-instance of $\Pi$, and
\item[(ii)] $|I'|+k'$ is bounded by~$f(k)$ for a computable function~$f$.
\end{itemize}
Respectively, $(I',k')$ is a \emph{kernel} and $f$ is its \emph{size}. A kernel is \emph{polynomial} if $f$ is polynomial.

While it can be shown that every decidable parameterized problem is \classFPT if and only if it admits a kernel, it is unlikely that every problem in \classFPT has a polynomial kernel (see, e.g., \cite{FominLSM19} for the details). Still, even if a paramterized problem admits no polynomial kernel up to some complexity conjectures, sometimes we can reduce it to solving of a polynomial number of instances of the same problem such that the size of each instance is bounded by a polynomial of the parameter. 

Let $\Pi$ be a parameterized problem and let $f\colon \mathbb{Z}^+ \rightarrow \mathbb{Z}^+$ be a computable function.  A \emph{Turing kernelization} or \emph{Turing kernel} for $\Pi$ of size $f$ is an algorithm that decides whether an instance $(I,k)$ of $\Pi$ is a yes-instance in time polynomial in $|I|+k$, when given access to an oracle that decides whether $(I',k')$ is a yes-instance of $\Pi$ in a single step if $|I'|+k\leq f(k)$.

It is typical to describe a compression or kernelization algorithm as a series of \emph{reduction rules}. A reduction rule is a polynomial algorithm that takes as an input an instance of a problem and output another, usually reduced, instance. A reduction rule is \emph{safe} if the input instance is a yes-instance if and only if the output instance is a yes-instance.

In our paper we consider kernelization/compression for weighted problems.  Because we are using weights, we have to compress their values as well. For this, we follow the  approach proposed by Etscheid et al.~\cite{EtscheidKMR17} that is based on the algorithm for compressing numbers given by Frank and Tardos in~\cite{FrankT87}. We state the result of Frank and Tardos in the form given in~\cite{EtscheidKMR17}.

\begin{proposition}[\cite{FrankT87}]\label{prop:compression}
There is an algorithm that, given a vector $w\in \mathbb{Q}^h$ and an integer $N$, in polynomial time finds a vector $\bar{w}\in\mathbb{Z}^h$ with $\|\bar{w}\|_{\infty}\leq 2^{4h^3}N^{h(h+2)}$ such that $\sign(w\cdot b)=\sign(\bar{w}\cdot b)$ for all vectors $b\in \mathbb{Z}^h$ with $\|b\|_1\leq N-1$.
\end{proposition}


\section{Subexponential algorithms for induced $d$-colorable subgraphs}\label{sec:subexpalgo}
To construct subexponential algorithms on $\chordalmke$, we consider tree decompositions such that each bag is ``almost'' a clique. 

\begin{definition}\label{def:almost}
Let $k$ be a nonnegative integer. 
A tree decomposition 
$\mathcal{T}=(T,\{X_t\}_{t\in V(T)})$ of a graph $G$ is \emph{$k$-almost chordal} if for every $t\in V(T)$, $\cliqc(G[X_t])\leq k$, that is, every bag can be converted to a clique by adding at most $k$ edges.
\end{definition}

Note that every chordal graph has $0$-almost chordal tree decomposition.  

Given a $k$-almost chordal tree decomposition, we are able to construct  dynamic programming algorithms that are subexponential in $k$ for various problems. The crucial property of the graphs in $\chordalmke$ is that we are able to construct nice $k$-almost chordal tree decompositions for them in  subexponential in $k$ time by making use of the following result of Fomin and Villanger~\cite{FominV13}. 

\begin{proposition}[\cite{FominV13}]\label{prop:FV12} 
Deciding whether graph $G$ is in $\chordalmke$ can be done in time $2^{\Oh(\sqrt{k}\log{k})} +\Oh(k^2nm)$. Moreover, if $G\in \chordalmke$, then the corresponding triangulation can be found in time $2^{\Oh(\sqrt{k}\log{k})} +\Oh(k^2nm)$. 
\end{proposition}

\begin{lemma}\label{lem:constr-almost}
A nice $k$-almost chordal decomposition of a graph $G\in\chordalmke$ with at most $n^2$ bags can be constructed in time $2^{\Oh(\sqrt{k}\log{k})}\cdot n^{\Oh(1)}$. 
\end{lemma}

\begin{proof}
Let $G$ be a graph. By Proposition~\ref{prop:FV12}, in time $2^{\Oh(\sqrt{k}\log{k})} +\Oh(k^2nm)$ we can construct a triangulation $H$ of $G$ 
such that  $H=G+A$ for $A\subseteq\binom{V(G)}{2}$ of size at most $k$.
Using Proposition~\ref{lem02:nice-tw}, we construct in polynomial time a nice tree decomposition $\mathcal{T}=(T,\{X_t\}_{t\in V(T)})$ of $H$ with at most $n^2$ bags such that 
every bag $X_t$ is a clique of $H$. Clearly, $\mathcal{T}$ is a nice tree decomposition of $G$ and $G[X_t]=H[X_t]-A$ for $t\in V(T)$. Therefore, $\cliqc(G[X_t])\leq k$ for $t\in V(T)$, that is, 
$\mathcal{T}$ is a  nice $k$-almost chordal decomposition of $G$.
\end{proof}

We need the following folklore observation whose proof is provided for completeness. 

\begin{lemma}\label{lem:deg-cliques}
A $d$-degenerate graph $G$ has at most $2^d\cdot n$ cliques  and all cliques of $G$ can be listed in  $2^d\cdot n^{\Oh(1)}$ time.
\end{lemma}

\begin{proof}
Let $G$ be a $d$-degenerate graph, and let $v_1,\ldots,v_n$ be a $d$-coloring ordering of its vertices.
Since $d_{G[v_i,\ldots,v_n]}(v_i)\leq d$, $G[v_i,\ldots,v_n]$ has at most $2^d$ cliques containing $v_i$ for every $i\in\{1,\ldots,n\}$. Therefore, $G$ has at most $2^d\cdot n$ cliques.  The cliques can be enumerated by brute force checking the subsets of the neighbors of $v_i$ in 
$G[v_i,\ldots,v_n]$ for every $i\in\{1,\ldots,n\}$. Because a $d$-coloring ordering can be found in polynomial (in fact, linear)
 time~\cite{MatulaB83}, the total running time is  $2^d\cdot n^{\Oh(1)}$. \end{proof}

The crux of our subexponential algorithms is in the following combinatorial lemma. 
\begin{lemma}\label{maincomblemma}
Let   $d\geq 1$ be an integer. 
Let $G$ be a 
 graph and let $\mathcal{F}$ be a set of induced $d$-colorable subgraphs of $G$. 
Let $U\subseteq V(G)$ be a set of vertices of $G$ such that $\cliqc(G[U])\leq k$, that is, $U$ can be made a clique by adding at most $k$ edges. 
Then 
\begin{itemize}
\item for every $F\in \mathcal{F}$, 
\[
|U\cap V(F)| \leq   \frac{3d+\sqrt{d^2 +8dk}}{2},   
\] and
\item 
the size of  the projection of $\mathcal{F}$ on $U$, that is, the size of the family 
\[\mathcal{S} =\{S \mid S=U\cap V(F) \text{ for some }F\in\mathcal{F} \}
\]
is at most
$(1+2^{(\sqrt{1+8k}-1)/2}\cdot |U|)^d$. 
\end{itemize}
Moreover, there is an algorithm that in time 
$2^{\Oh(d\sqrt{k})}\cdot n^{\Oh(d)}$ outputs a family of sets $\mathcal{S}'\supseteq  \mathcal{S}$ such that  each set from $\mathcal{S}'$ has 
 at most $ \frac{3d+\sqrt{d^2 +8dk}}{2}$ vertices, the number of sets in $\mathcal{S}'$  is  $(1+2^{(\sqrt{1+8k}-1)/2}\cdot n)^d$
 and  $G[S]$ is $d$-colorable for $S\in\mathcal{S}'$.
\end{lemma}

\begin{proof} 
We partition  $U$ into sets $X$ and  $Y$ as follows. Let
 $X$ be the vertices of $U$ that have at least one non-neighbor in $U$. In other words, for every $v\in X$ there is $u\in U$ that is not adjacent to $v$. Two observations about set $X$ will be useful. First, because $U$, and hence $X$, can be turned into a clique by adding at most $k$ edges, we have that $|X|\leq 2k$. Second, the remaining vertices of $U$, namely, $Y=U\setminus X$, form a clique.  For 
every set $S\in \mathcal{S}$, we define $S_X=X\cap S$ and  $S_Y=Y\cap S$. Note that $S=S_X\cup S_Y$.

Because $Y $ is a clique in $G$, no $d$-colorable subgraph from $\mathcal{F}$ can contain more than $d$ vertices from  $Y$. Hence, $|S_Y|\leq d$.

Let $x=|S_X|$. Because  $G[S_X]$ is an induced subgraph of some 
$d$-colorable graph $F\in \mathcal{F}$, we have that 
$G[S_X]$ is 
 $d$-colorable.  On the other hand, since  $\cliqc(G[U])\leq k$, $G[S_X]$ can be turned into  complete graph by adding at most $k$ edges. These two conditions are used to  estimate $x$. Let us recall that 
\emph{Tur\'an graph}  is the complete $d$-partite graph on $x$ vertices whose partition sets differ in size by at most $1$.  
According to Tur\'an's  theorem, see e.g. \cite{Diestel12}, Tur\'an graph has the maximum possible number of edges among all $d$-colorable graphs. The number of edges in Tur\'an's graph is at most $\frac{1}{2}x^2 \frac{d-1}{d}$. Thus, 
\[
\binom{x}{2}-k\leq |E(G[S_X])|\leq \frac{1}{2}x^2 \frac{d-1}{d}
\]
and
\[
k\geq \binom{x}{2} -\frac{1}{2}x^2 \frac{d-1}{d}=\frac{x^2-dx}{2d}.
\]
Therefore, \[x\leq \frac{d+\sqrt{d^2 +8dk}}{2}.
\]

We obtain that 
\[
|S|=|S_X|+|S_Y|\leq x +d\leq \frac{3d+\sqrt{d^2 +8dk}}{2},
\]
which proves the first claim of the lemma. 

To prove the second claim, let $H=G[U]$. Observe that $\overline{H}$, the complement  of $H$, has at most $k$ edges. Consider $Z\subseteq V(H)$. If $|Z|\leq \frac{\sqrt{1+8k}+1}{2}$, then the minimum degree $\delta(\overline{H}[Z])\leq \frac{\sqrt{1+8k}-1}{2}$. If $|Z|>\frac{\sqrt{1+8k}+1}{2}$, then
$$\delta(\overline{H}[Z])\leq \frac{2|E(\overline{H}[Z])|}{|Z|}\leq \frac{4k}{\sqrt{1+8k}+1}=\frac{\sqrt{8k+1}-1}{2},$$
that is, the minimum degree of every induced subgraph of $\overline{H}$ is at most $\frac{\sqrt{8k+1}-1}{2}$.
Therefore,  $\overline{H}$ is $\frac{\sqrt{1+8k}-1}{2}$-degenerate. 

An induced subgraph of $H$ with the vertex set $S$ is $d$-colorable if and only if $S$  can be partitioned into at most $d$ independent sets. Equivalently, 
$H[S]$ is $d$-colorable if and only if $S$ can be partitioned into at most $d$ cliques of $\overline{H}$. By Lemma~\ref{lem:deg-cliques}, $\overline{H}$
has at most $2^{(\sqrt{8k+1}-1)/2} \cdot |U|$ cliques. Then $U$ contains at most $(1+2^{(\sqrt{1+8k}-1)/2}\cdot |U|)^d$ subsets $S$ such that $S$ can  
 be partitioned into at most $d$ cliques of $\overline{H}$. 
 
 To complete the proof, observe that the cliques of $\overline{H}$ can be listed in time $2^{(\sqrt{1+8k}-1)/2}\cdot n^{\Oh(1)}$ by Lemma~\ref{lem:deg-cliques}.
 Then in  $2^{\Oh(d\sqrt{k})}\cdot n^{\Oh(d)}$ time, we construct $\mathcal{S}'$ by considering the unions of $d$ cliques;  the cliques can be the same or empty. 
\end{proof}

 Let $G$ be a graph and let $F$ be an induced $d$-colorable subgraph of $G$.  Informally, Lemma~\ref{maincomblemma} says that for a given a $k$-almost chordal tree decomposition,  every bag of this tree decomposition contains roughly $\Oh(d+\sqrt{dk})$ vertices of $F$. This statement combined with dynamic programming over the tree decomposition  could easily bring us to the algorithm computing  a maximum $d$-colored subgraph of $G$ in time  
 $n^{\Oh(d+\sqrt{dk})}$. However, this is not what we are shooting for;  such an  algorithm is not fixed-parameter tractable with parameter $k$. This is where the second part of the lemma becomes extremely helpful. Let us look at the family of all $d$-colorable induced subgraphs of $G$. Then the number of different intersections of the graphs from this family with a single bag of the tree decomposition is  bounded by 
 $2^{\Oh(d\sqrt{k})}\cdot n^{\Oh(d)}$.
 This allows us to bound the number of ``partial solutions'' in the dynamic programming, which in turn brings us to a parameterized subexponential algorithm. As an example of the applicability of Lemma~\ref{maincomblemma}, we  give an algorithm for the following generic problem. 
 
\defproblema{\probWdcS}{Graph $G$ with weight function $w\colon V(G) \to \Bbb{R}$. }{Find a  $d$-colorable induced subgraph $H$ of $G$ of maximum weight $\sum_{v\in V(H)} w(v)$.}

For $d=1$, this is the problem of finding an independent set of maximum weight, the \probWIS problem. For $d=2$, this is the problem of finding an induced bipartite subgraph of maximum weight,  \probWBS.

We prove the following theorem for \probWdcS. Note that in Theorem~\ref{theorem-dcolor} we do not require that the input graph is in \chordalmke.

\begin{theorem}\label{theorem-dcolor} Let  $d\geq 1$ be an integer. For a given 
  graph $G$ with a  nice $k$-almost chordal tree decomposition with $n^{\Oh(1)}$ bags, the 
\probWdcS  problem is solvable in time 
$2^{\Oh(\sqrt{k}\cdot d\log d)} \cdot n^{\Oh(d)}$. 
\end{theorem}

\begin{proof}
Let $\mathcal{T}=(T,\{X_t\}_{t\in V(T)})$ be a nice $k$-almost chordal tree decomposition of $G$ with $|V(T)|=n^{\Oh(1)}$.
 We perform dynamic programming over $\mathcal{T}$. Let us note that the width of the  decomposition can be of order of $n$. The proof of the correctness for  this dynamic programming is very similar to the one provided normally for graphs of bounded treewidth. However the running time analysis is based on Lemma~\ref{maincomblemma}. 

Recall that $T$ is rooted at some node $r$, and for a node $t$ of $T$, $V_t$ denotes  the union of all the bags present in the subtree of $T$ rooted at $t$, including $X_t$. For vertex sets $X\subset X'$ of graph $G$, we say that a coloring $c$ of $G[X]$ is \emph{extendible} to a coloring $c'$ of $G[X']$, if for every $x\in X$, $c(v)=c'(v)$. 

 For every node $t$, every $S\subseteq X_t$ such that $G[S]$ is $d$-colorable, every mapping $c\colon S \to \{1, \dots, d\}$ of $G[S]$, we define the following value:
\begin{eqnarray}\label{def_of cost}
\cost[t,S,c] & = & {\textrm{maximum possible weight of a set $\widehat{S}$ such that}} \\
& & {\textrm{$S\subseteq \widehat{S}\subseteq V_t$, $\widehat{S}\cap X_t=S$,  and  $c$ is a proper coloring of $G[S]$  extendible}}\nonumber \\ 
& &
{\textrm{to a proper $d$-coloring of $G[\widehat{S}]$.\nonumber}} 
\end{eqnarray}

If $c$ is not a proper coloring of $G[S]$ or if  no such set $\widehat{S}$ exists, then we put $\cost[t,S,c]=-\infty$. 
We also put $\cost[t,\emptyset,c]$ be the maximum possible weight of a set $\widehat{S}$ such that  
 $\widehat{S}\subseteq V_t$, $\widehat{S}\cap X_t=\emptyset$, and  $G[\widehat{S}]$ is $d$-colorable.
Then $\cost[r,\emptyset,c]$ is exactly the maximum weight of a $d$-colorable induced subgraph   in $G$; this is due to the fact that $V_r=V(G)$ and $X_r=\emptyset$.

\medskip\noindent {\bf{Leaf node.}} If $t$ is a leaf node, then we have    $\cost[t,\emptyset,c]=0$.
In this case, because the leaf node is an empty bag, the correctness of the formula is trivial.

\medskip\noindent {\bf{Introduce node.}} Suppose $t$ is an introduce node with child $t'$ such that $X_t=X_{t'}\cup \{v\}$ for some $v\notin X_{t'}$. Let $S$ be any subset of $X_t$. If $c$ is not a proper    $d$-coloring of $G[S]$,   we   put $\cost[t,S,c]=-\infty$. Otherwise, we claim that the following formula holds: 
\begin{eqnarray}\label{eq:i}
\cost[t,S,c] & = & \begin{cases} \cost[t',S,c] & \textrm{ if $v\notin S$;}\\
 \cost[t',S\setminus \{v\},c']+\wei(v) & \textrm{ otherwise.} \end{cases}
\end{eqnarray}
Here $c'$ is the coloring of $S\setminus \{v\}$ extendible to  $c$. 

\begin{claim}\label{claim:introd} Formula \eqref{eq:i} is correct. 
\end{claim}

\begin{subproof}[Proof of Claim~\ref{claim:introd}]
When $v\not\in S$,  
 the formula trivially holds. Suppose that  $v \in S$. 
Let $\widehat{S}$ be a set maximizing  $\cost[t,S,c]$. Because coloring $c'$ of $S\setminus \{v\}$ is extendible to  coloring $c$ of ${S}$, which in turn is extendible to a $d$-coloring of $\widehat{S}$, we have that 
 set  $\widehat{S}\setminus \{v\}$ is  one of the sets considered in the definition of $\cost[t',S\setminus \{v\}, c']$. 
 Hence $\cost[t',S\setminus \{v\}, c']\geq \wei(\widehat{S}\setminus \{v\})=\wei(\widehat{S})-\wei(v)=\cost[t,S,c]-\wei(v)$.

 On the other hand, let $\widehat{S}'$ be a set for which the maximum is attained in the definition of $\cost[t',S\setminus \{v\},c']$. By the property of tree decompositions, all neighbors  of $v$ in $V_t$, are contained in the bag $X_t$. This yields that set  $\widehat{S}' \cup \{v\}$ induces a graph whose $d$-coloring can be obtained by extending coloring $c$ of $ {S} $, therefore, 
$\cost[t,S,c]\geq \wei(\widehat{S}'\cup \{v\})=\cost[t',S\setminus \{v\},c']+\wei(v) $. We conclude that
$\cost[t,S,c]=\cost[t',S\setminus \{v\},c']+\wei(v) $.
\end{subproof}

To evaluate the running time required to compute  \eqref{eq:i}, note that for fixed $S$ and $c$, we have to verify whether $c$ is a proper coloring of $S$, which can be done in time $\Oh(|S|^2)$ by going through all adjacencies of $G[S]$.

\medskip\noindent {\bf{Forget node.}} Let $t$ be a forget node with child $t'$ such that $X_t=X_{t'}\setminus \{w\}$ for some $w\in X_{t'}$. Let $S$ be any subset of $X_t$; again we assume that $c$ is a proper    $d$-coloring of $G[S]$, since otherwise we put $\cost[t,S, c]=-\infty$. We claim that the following formula holds:
\begin{eqnarray}\label{eq:f}
\cost[t,S,c] & = & \max\Big\{\cost[t',S,c],\max_{c \text{ is extendible to } c'}\cost[t',S\cup \{w\},c']\Big\}.
\end{eqnarray}

\begin{claim}\label{claim:forger} Formula \eqref{eq:f} is correct. 
\end{claim}

\begin{subproof}[Proof of Claim~\ref{claim:forger}]
Because $V_t=V_{t'}$, we have that 
\begin{eqnarray*}
\cost[t,S,c] & 
\geq  & \max\Big\{\cost[t',S,c],\max_{c \text{ is extendible to } c'}\cost[t',S\cup \{w\},c']\Big\}.
\end{eqnarray*}
On the other hand, 
  let $\widehat{S}$ be a set for which the maximum is attained in the definition of $\cost[t,S,c]$. If $w\notin \widehat{S}$, then $\cost[t,S,c] =\cost[t',S,c]$.
  If $w\in \widehat{S}$, then \[\cost[t,S,c] =\max_{c \text{ is extendible to } c'}\cost[t',S\cup \{w\},c'].\] 
  Thus 
  \begin{eqnarray*}
\cost[t,S,c] & 
\leq  & \max\Big\{\cost[t',S,c],\max_{c \text{ is extendible to } c'}\cost[t',S\cup \{w\},c']\Big\}.
\end{eqnarray*}
\end{subproof}
For running time, to compute \eqref{eq:f}, one has to go through all colorings $c'$ of $S\cup \{w\}$ which are extensions of coloring $c$ of $S$. Thus for every possible color $i$ of $w$, we have to check if this color is  compatible with the coloring by $c$  the neighbors of $w$ in $S$. This can be done in time $\Oh(|S|)$.  

\medskip\noindent {\bf{Join node.}} Finally, suppose that $t$ is a join node with children $t_1,t_2$ such that $X_t=X_{t_1}=X_{t_2}$. Let $S$ be any subset of $X_t$; as before, we can assume that    $c$ is a proper    $d$-coloring of $G[S]$. The claimed recursive formula is as follows:
\begin{eqnarray}\label{eq:j}
\cost[t,S,c] & = & \cost[t_1,S,c]+\cost[t_2,S,c]-\wei(S).
\end{eqnarray}

\begin{claim}\label{claim:join} Formula \eqref{eq:j} is correct. 
\end{claim}

\begin{subproof}[Proof of Claim~\ref{claim:join}]
Let  $\widehat{S}$   be a set for which the maximum is attained in the definition of $\cost[t,S,c]$ and   Let $\widehat{S}_1 = \widehat{S}\cap V_{t_1}$ and $\widehat{S}_2 = \widehat{S}\cap V_{t_2}$. Because 
$\widehat{S}\supseteq \widehat{S}_i$, we have that 
$\cost[t_i,S,c]\geq \wei(\widehat{S}_i)$,  $i=1,2$. This yields that 
\[\cost[t_1,S,c] +\cost[t_2,S,c] \geq \wei(\widehat{S}_1)+\wei(\widehat{S}_2)=\wei(\widehat{S})- \wei(\ {S})=\cost[t,S,c]. 
\]
On the other hand, the union of sets maximizing the  costs of  $\cost[t_1,S,c]$ and $\cost[t_2,S,c]$, is a set whose $d$-coloring can be obtained from extending coloring $c$ of $S$. Thus 
\[\cost[t_1,S,c] \geq \cost[t_1,S,c] +\cost[t_2,S,c]  - \wei(\ {S}). 
\]
\end{subproof}
The running time to compute   \eqref{eq:j} for fixed set $S$ and coloring $c$, is proportional to the time required to checked whether $c$ is a proper coloring. Again, this can be done in time $\Oh(|S|^2)$.

\medskip

This concludes the description and the proof of correctness of the recursive formulas for computing the values of $\cost[\cdot,\cdot,\cdot]$. The optimal subgraph of the maximum weight can be found by the standard backtracking arguments. Let us now  estimate the total running time.  The running time of our dynamic programming algorithm, up to a   multiplicative factor $\Oh(|S|^2)$, is  dominated by the number of triples $[t,S,c]$.  
The number $t$ is in $n^{\Oh(1)}$. Every set  $S$ should   induce a $d$-colorable subgraph, so we can restrict our attention only to sets of the form $X_t\cap V(F)$ for some $d$-colorable graph $F$. By Lemma~\ref{maincomblemma}, 
each of these sets is of size at most  $d+\frac{d+\sqrt{d^2 +8dk}}{2} $ and the total number of such  sets $S$ for each bag $X_t$ is 
 is $2^{\Oh(d\sqrt{k})} \cdot n^{\Oh(d)}$ and they can be listed in $2^{\Oh(d\sqrt{k})} \cdot n^{\Oh(d)}$  time.
  Thus, the number of $d$-colorings $c$ of each of the sets $S$ is  $d^{\Oh(|S|)}=d^{\Oh(d+\sqrt{dk})}$. Hence the total running time of the  dynamic programming is  
$2^{\Oh(\sqrt{k}\cdot d\log d)} \cdot n^{\Oh(d)}$.
\end{proof}

Combining Proposition~\ref{lem02:nice-tw},   Lemma~\ref{lem:constr-almost} and Theorem~\ref{theorem-dcolor}, we immediately obtain the following corollary.
Recall that $A\subseteq \binom{V(G)}{2}\setminus E(G)$ is a \chordal-modulator if $G+A$ is a chordal graph.

\begin{corollary}\label{cor:dcolor} \probWdcS on  a 
graph $G\in \chordalmke$ 
is solvable in time $2^{\Oh(\sqrt{k}(\log k+d\log d))} \cdot n^{\Oh(d)}$. Moreover, the problem can be  solved in $2^{\Oh(\sqrt{k}\cdot d\log {d})} \cdot n^{\Oh(d)}$ time
if a \chordal-modulator of size at most $k$ is given. 
\end{corollary}

By Corollary~\ref{cor:dcolor}, we immediately derive the following. 
\begin{corollary}\label{cor:WIS}
\probWIS and  \probWBS on  $G\in\chordalmke$ are solvable in time  $2^{\Oh(\sqrt{k}\log{k})} \cdot n^{\Oh(1)}$. 
Moreover, the problems can be solved in $2^{\Oh(\sqrt{k})} \cdot n^{\Oh(1)}$ time if a \chordal-modulator of size at most $k$ is given. 
\end{corollary}
 
 In the \probWVC, we are given a weighted graph $G$ and the task is to find a vertex cover of minimum weight, that is, a set of vertices $X$ such that every edge of $G$ has at least one endpoint in $G$.  Similarly, in  the  \probWOCT, we are asked to find a set of vertices of minimum weight such that every cycle of odd length contains at least one vertex from the set. 
Since the  complement of every independent set  is a vertex cover, and the complement of every induced bipartite subgraph
is an odd cycle transversal, we have the following corollary.  
\begin{corollary}\label{cor:VC}
\probWVC and \probWOCT on  $ \chordalmke$ graphs are  solvable in time  $2^{\Oh(\sqrt{k}\log{k})} \cdot n^{\Oh(1)}$.
Moreover, the problems can be solved in $2^{\Oh(\sqrt{k})} \cdot n^{\Oh(1)}$ time if a \chordal-modulator of size at most $k$ is given. 
\end{corollary}


\subsection{Various extentions}\label{subs:ext}
The technique developed to prove 
Theorem~\ref{theorem-dcolor} could be used to obtain subexponential algorithms for other problems beyond  \probWdcS. These algorithms are very similar to the one 
from Theorem~\ref{theorem-dcolor} and we sketch here only few problems and the adjustments   required to make the dynamic programming from Theorem~\ref{theorem-dcolor} to be applicable to these problems. 
 
A {\em homomorphism} $G\to H$ from a graph $G$ to a  graph $H$ is a mapping
from the vertex set of $G$ to that of $H$ such that the image of every edge of $G$ is an edge of $H$.  
In other words, a homomorphism $G\to H$ exists if and only if there is a mapping $g : V(G) \to V(H)$,
 such that for every edge $uv \in E(G)$,  we have $g(u)g(v) \in E(H)$. 
 Since there is a homomorphism from $G$ to a complete graph $K_d$ on $d$ vertices if and only if $G$ is $d$-colorable, the deciding whether there is a homomorphism from $G$ to $H$ is often referred as the $H$-coloring. Note that if $G$ admits an $H$-coloring, then $G$ is $|V(H)|$-colorable.
 
 The main difference between solving \probWHcS, the problem of finding the maximum weight induced subgraph admiting a homomorphism to $H$,  with   Theorem~\ref{theorem-dcolor}  is that the value $\cost[t,S,c]$ in \eqref{def_of cost} should be redefined.  
 In order to find a maximum weight $H$-colorable induced subgraph of a graph $G$, we need to compute  
 \begin{eqnarray*} 
\cost[t,S,g] & = & {\textrm{maximum possible weight of a set $\widehat{S}$ such that}} \\
& & {\textrm{$S\subseteq \widehat{S}\subseteq V_t$, $\widehat{S}\cap X_t=S$,  and  $g$ is a homomorphism from  $G[S]$  to $H$ extendible}}\nonumber \\ 
& &
{\textrm{to a homomorphism of $G[\widehat{S}]$ to $H$.\nonumber}} 
\end{eqnarray*}
 The number of homomorphisms $g$ from $G[S]$ to $H$ does not exceed $|V(H)|^{|S|}$. The running time of operations  join, introduce and forget operations in the dynamic programming can be bounded by a polynomial of the  number of states. Hence we can solve the problem in time 
 $2^{\Oh(\sqrt{k}(\log k+|V(H)|\log |V(H)|)}\cdot n^{\Oh(|V(H)|)}$. 
 
 Similar running times could be derived for the variants of   \probWdcS where some additional constrains on the properties of the $d$-colorable induced subgraph of minimum weight are imposed by some property $\mathcal{C}$. For example, property $\mathcal{C}$ could be that  the required subgraph is connected, acyclic, regular, degenerate, etc. As far as the information of the partial solution required for  property $\mathcal{C}$ is characterized by set  $S\subseteq V_t$ and all possible subsets of $S$ or all permutations of $S$, we can solve the corresponding problem in time   
  $2^{\Oh((d\sqrt{k})\log(dk))}\cdot n^{\Oh(d)}$. 
 
 As a concrete example, consider \probWddS, whose task is to find a maximum weight $d$-degenerate induced subgraph $H$ of the input graphs $G$.  
Recall that a graph is $d$-degenerate for a nonnegative integer $d$, if every induced subgraph of $G$ has a vertex of degree at most $d$. In particular, every forest  is   $1$-degenerate. Recall also that 
a graph $G$ is $d$-degenerate if and only if it admits an $d$-coloring ordering~\cite{MR0193025}. In particular, this immediately implies that every $d$-degenerate graph is $(d+1)$-colorable, that is, the chromatic number bound is automatically given by the $d$-degeneracy property. 

Let $G$ be a graph and let $\pi$ be an ordering of its vertices. Let $X\subseteq V(G)$.  We say that $\pi$ as an \emph{extension} of an ordering $\pi\rq{}$ of $X$ if the vertices of $X$ occur in $\pi$ in the same order as in $\pi'$. Suppose that $X=\{x_1,\ldots,x_r\}$ and the vertices of $X$ are indexed with respect to their order in $\pi$. Denote by $\delta_\pi(G,X)=(\delta_1,\ldots,\delta_r)$ the sequence of nonnegative integers such that for each $i\in\{1,\ldots,k\}$, $\delta_i$ is the degree of $x_i$ in the graph obtained from $G$ by the deletion of the vertices that occur before $x_i$ in $\pi$. Notice that if $\pi$ is a $d$-coloring ordering of $G$, then $\delta_i\leq d$ for $i\in\{1,\ldots,r\}$. 
We use the following lemma to construct our algorithm for \probWddS.

\begin{lemma}\label{lem:degenerate}
Let $G$ be a graph and let $V_1,V_2\subseteq V(G)$ such that $V_1\cup V_2=V(G)$. Let also $X=V_1\cap V_2$ with $r=|X|$. Then $G$ has a $d$-coloring ordering if and only if there is  an ordering $\pi$ of $X$ and orderings $\pi_1$ and $\pi_2$ of $V_1$ and $V_2$ respectively such that 
\begin{itemize}
\item[(i)] $\pi_i$ is an extension of $\pi$ for $i=1,2$,
\item[(ii)] for the sequences $\delta_{\pi}(G[X],X)=(\delta_1,\ldots,\delta_r)$,
 $\delta_{\pi_1}(G[V_1],X)=(\delta_1^{(1)},\ldots,\delta_r^{(1)})$ and \linebreak  $\delta_{\pi_2}(G[V_2],X)=(\delta_1^{(2)},\ldots,\delta_r^{(2)})$, it holds that 
 $\delta_i^{(1)}+\delta_i^{(2)}-\delta_i\leq d$ for every $i\in\{1,\ldots,r\}$.
\end{itemize}
\end{lemma}

\begin{proof}
Suppose that $\pi\rq{}$ is a $d$-coloring ordering. Then we define $\pi$, $\pi_1$ and $\pi_2$ as the orderings of $X$, $V_1$ and $V_2$ respectively as the ordering inherited from $\pi$, that is, every two vertices of the corresponding sets occur exactly in the same order as in $\pi$. It is straightforward to  verify that (i) and (ii) are fulfilled. For the opposite direction, assume that there are $\pi$, $\pi_1$ and $\pi_2$ satisfying (i) and (ii). Let $\pi=(x_1,\ldots,x_r)$. We construct the ordering $\pi\rq{}$ of $G$ as follows. First, we concatenate the suborderings of $\pi_1$ and $\pi_2$ containing the vertices that occur before $x_1$ in these orderings and add $x_1$ in the end. Then for $i=2,\ldots,r$, we consecutively add suborderings of $\pi_1$ and $\pi_2$ composed by the vertices that occur  between $x_{i-1}$ and $x_i$ and add $x_i$ in the end. Finally, we add suborderings of $\pi_1$ and $\pi_2$ with the vertices that are after $x_r$. It is easy to see that $\pi\rq{}$ is a $d$-coloring ordering of $G$ because of (i) and (ii).  
\end{proof}

Using Lemma~\ref{lem:degenerate}, we can define the values computed by our dynamic programming algorithm for  \probWddS.
 For a nonnegative integer $r$, we define 
$\Delta_r=\{(\delta_1,\ldots,\delta_r)\mid 0\leq \delta_i\leq d,~1\leq i\leq r\}$, i.e, $\Delta_r$ is the set of all sequences of length $r$ of nonnegative integers at most $d$. 
Similarly to \eqref{def_of cost}, for every node $t$, every $S\subseteq X_t$ such that $G[S]$ is $(d+1)$-colorable, every ordering $\pi$ of $S$ and every $\delta\in \Delta_{|S|}$, we define
  \begin{eqnarray*} 
\cost[t,S,\pi,\delta] & = & {\textrm{maximum possible weight of a set $\widehat{S}$ such that}} \\
& & {\textrm{$S\subseteq \widehat{S}\subseteq V_t$, $\widehat{S}\cap X_t=S$,  and  ordering $\pi$ of  $S$ can be extended }}  \\ 
 & &
 {\textrm{to a $d$-coloring ordering $\pi\rq{}$ of $G[\widehat{S}]$ with $\delta_{\pi\rq{}}(G[\widehat{S}],S)=\delta$.}} 
\end{eqnarray*} 
For a given $S\subseteq X_t$, there are $|S|!$ different ordering of the vertices of $S$ and $|\Delta_{|S|}|=(d+1)^{|S|}$. By Lemma~\ref{maincomblemma}, we have that $|S|\leq \frac{3d+\sqrt{d^2 +8dk}}{2}$ and, therefore, we consider $2^{\Oh((d+\sqrt{dk})\log(dk))}$ orderings $\pi$ and $2^{\Oh((d+\sqrt{dk})\log d)}$ sequences $\delta$. Taking into account that by Lemma~\ref{maincomblemma}, we consider $2^{\Oh(d\sqrt{k})} \cdot n^{\Oh(d)}$ sets $S$, we obtain that the table of our dynamic programming algorithm stores 
$2^{\Oh((d\sqrt{dk})\log (kd))} \cdot n^{\Oh(d)}$ values. 
Applying Lemma~\ref{lem:degenerate} and the standard dynamic programming arguments, it is easy to show that we can solve \probWddS using the information stored in the tables.

 Many natural problem can be described by combining  connectivity, degeneracy or degree constraints. Examples of such problems are various maximization problems like 
\textsc{Weighted Induced Tree}, 
  \textsc{Weighted Induced Path},  \textsc{Weighted Induced Cycle}, as well as various packing variants of these problems like   \probWICP, the problem of finding a maximum induced subgraph whose each connected component is a  cycle.  
 
We summarize these observations with the following theorem. 

 \begin{theorem}\label{thm:general}
Let  $d\geq 1$ be an integer and $G$ be a 
graph from  $\chordalmke$. Then 
\begin{itemize}
\item \probWHcS  can be solved in  $2^{\Oh(\sqrt{k}(\log k+|V(H)|\log |V(H)|))}\cdot n^{\Oh(|V(H)|)}$ time, 
\item  \probWddS is solvable in time $2^{\Oh((d\sqrt{k})\log(dk))}\cdot n^{\Oh(d)}$,
\end{itemize}
and
\begin{itemize}
\item \probmWIF (\probWFVS ),
\item \textsc{Weighted Induced Tree},
 \item \textsc{Weighted Induced Path (Cycle)}
\item  \probWICP,
\end{itemize}
are solvable in $2^{\Oh(\sqrt{k}\log k )}\cdot n^{\Oh(1)}$ time.
\end{theorem}

In some cases, we can obtain a better running time if a \chordal-modulator of size at most $k$ is given. For  \probWHcS, this is done in the same way as for \probWdcS. For some other problems, like   \probmWIF (\probWFVS ), this would demand using recent techniques for dynamic programming on graphs of bounded treewidth for problems with connectivity contstraints (see~\cite{cut-and-count,BodlaenderCKN15,FominLPS16,DBLP:journals/corr/abs-1104-3057})
but this goes beyond the scope of our paper. 

Another extension of Theorem~\ref{theorem-dcolor} can be derived from the very recent results of Baste, Sau and Thilikos~\cite{BasteST20} about the \textsc{$\mathcal{F}$-Minor Deletion} problem on graphs of bounded treewidth. Recall that a graph $F$ is a \emph{minor} of $G$ if a graph isomorphic to $F$ can be obtained from $G$ by vertex and edge deletions and edge contractions.  
Respectively, $G$ is said to be \emph{$F$-minor free} if $G$ does not contain $F$ as a minor. For a family of graphs $\mathcal{F}$, $G$ is $\mathcal{F}$-minor free if $G$ is $F$-minor free for every $F\in\mathcal{F}$. For a family $\mathcal{F}$, the task of \textsc{$\mathcal{F}$-Minor Deletion} is, given a graph $G$, to find a minimum set of vertices $X$ such that $G-X$ is $\mathcal{F}$-minor free. Clearly, \textsc{$\mathcal{F}$-Minor Deletion} is equivalent to \textsc{$\mathcal{F}$-Minor Free Induced Subgraph}, whose task is to find a maximum $\mathcal{F}$-minor free induced subgraph of $G$. A family of graphs $\mathcal{F}$ is \emph{connected} if every $F\in\mathcal{F}$ is a connected graph. Baste et al.~\cite{BasteST20} obtained, in particular, the following result. 

\begin{proposition}[\cite{BasteST20}]\label{prop:minor-free}
Let $\mathcal{F}$ be a finite connected family of graphs. Then \textsc{$\mathcal{F}$-Minor Deletion} can be solved in time $2^{\Oh(w \log w)}\cdot n^{\Oh(1)}$ on graphs of treewidth at most $w$.\footnote{the constants hidden in the big-$\Oh$ notation depend on $\mathcal{F}$.}
\end{proposition}

It is well-known (see, e.g., the book~\cite{NesetrilM12} for the inclusion relations between the classes of sparse graphs) that if $\mathcal{F}$ is a finite family, then there is a positive integer $d$ such that every $\mathcal{F}$-minor free graph is $d$-degenerate. This means that for a finite family  $\mathcal{F}$,  $\mathcal{F}$-minor free graphs are $d$-colorable for some constant $d$ that depends on $\mathcal{F}$ only. This allows us to use Lemma~\ref{maincomblemma} and then combine our approach from Theorem~\ref{theorem-dcolor} with the techniques of Baste et al.~\cite{BasteST20}. Using Lemma~\ref{prop:minor-free}, we obtain the following theorem.

 \begin{theorem}\label{thm:minor-free}
Let $\mathcal{F}$ be a finite connected family of graphs. 
Let  also $G$ be a 
 graph from  $\chordalmke$. Then 
 \textsc{$\mathcal{F}$-Minor Free Induced Subgraph} (or, equivalently, \textsc{$\mathcal{F}$-Minor Deletion}) can be solved
in $2^{\Oh(\sqrt{k}\log k)}\cdot n^{\Oh(1)}$ time.
\end{theorem}

For example, this framework encompasses such problems as \textsc{Induced Planar Subgraph} or \textsc{Induced Outerplanar Subgraph} whose task is to find a  subgraph 
of maximum size that is planar or outerplanar, respectively.

With a small adjustment the dynamic programming could be applied to the problems with specific requirements on the complement of the maximum induced $d$-colored subgraph. 
 For example,  consider the following problem. A set of vertices $S\subseteq V(G)$ is a \emph{connected vertex cover} if $S$ is a vertex cover and $G[S]$ is connected. Then in the \probWCVC problem, we are  given a graph $G$ with a weight function $w\colon V(G)\rightarrow \mathbb{Z}^+$ and the task is to find a connected vertex cover in $G$ of minimum weight. 
  Similarly,  \probWCFVS is the problem of finding a connected feedback vertex set of minimum weight. 
  
  The complement of every vertex cover is an independent set, that is a $1$-colorable subgraph, and the complement of every feedback vertex set is a forest, hence $2$-colorable subgraph.  While now the connectivity constraints are not on the maximum induced subgraph but on its complement our previous arguments can be adapted to handle these problems. 

\begin{theorem}\label{thm:WCVC}
\probWCVC and  \probWCFVS are  solvable in time $2^{\Oh(\sqrt{k}\log k)}\cdot n^{\Oh(1)}$ on $\chordalmke$.
\end{theorem}

\begin{proof}[Proof sketch]
The proof is similar to the proof of  Theorem~\ref{theorem-dcolor}, so we only briefly sketch the main idea. We also give the sketch for \probWCVC only, solution to \probWCFVS is basically the same.

Let $(G,w,W)$ be an instance of \probWCVC.

We use Lemma~\ref{lem:constr-almost} to find a nice $k$-almost chordal tree decomposition $\mathcal{T}=(T,\{X_t\}_{t\in V(T)})$ of $G$ with $\Oh(n^2)$ bags.  Let $r$ be the root of $T$ and $V_t$ be the union of the bags in the subtree of $T$ rooted in $t$. 

The following claim is crucial for our algorithm.

\begin{claim}\label{cl:comp}
Let $t$ be a node of $\mathcal{T}$ and let $S\subseteq X_t$. Then $G[S]$ has at most $ \frac{3+\sqrt{1+8k}}{2}$ connected components. 
\end{claim}

To show the claim, it is sufficient to observe that any set of vertices constructed by picking an arbitrary vertex from each connected component of $G[S]$ is independent. Hence, the number of connected components of $G[S]$ is upper bounded by the maximum size of an independent set in $X_t$.  By Lemma~\ref{maincomblemma}, any independent subset of $X_t$ has size at most $ \frac{3+\sqrt{1+8k}}{2}$ and the claim follows.  

\medskip
Let $t\in\mathcal{T}$. For a set $S\subseteq X_t$, denote by $\mathcal{P}_t(S)$ the set of all partitions of $S$ such that every two vertices $x,y\in S$ that are in the same component are in the same set of each partition; we assume that $\mathcal{P}_t(\emptyset)=\{\emptyset\}$.

 Notice that $S\subseteq V(G)$ is a connected vertex cover if and only if $U=V(G)\setminus S$ is an independent set and $G[S]$ is a connected graph.

For every $t\in V(T)$, every independent set $U\subseteq X_t$ and every  $P\in\mathcal{P}_t(S)$ for $S=X_t\setminus U$, our algorithm computes  the value 
$\cost[t,U,P]$ that is the minimum possible weight of a set $\widehat{S}\subseteq V_t$ such that 
\begin{itemize}
\item $\widehat{S}\cap X_t=S=X_t\setminus U$,
\item $\widehat{S}$ is a vertex cover of $G[V_t]$,
\item for each connected component $C$ of $G[\widehat{S}]$, $V(C)\cap S\neq\emptyset$ unless $S=\emptyset$,
\item every two vertices $x,y\in S$ are in the same set of $P$ if and only if they belong to the same connected component of $G[\widehat{S}]$.
\end{itemize}

We compute $\cost[t,U,P]$ bottom-up for $t\in V(T)$ starting from the leaves. The computations are trivial for leaf nodes and are performed in a standard way for introduce, forget and join nodes. Finally, we compute $\cost[r,\emptyset,\{\emptyset\}]$, which  gives us the minimum weight of a connected vertex cover. 

By Lemma~\ref{maincomblemma}, $X_t$ has $2^{\Oh(\sqrt{k})}\cdot n^{\Oh(1)}$ independent subsets. 
By Claim~\ref{cl:comp}, we have that for every $S\subseteq X_t$, $G[S]$ has at most $\frac{3+\sqrt{1+8k}}{2}$ connected components. Therefore, $|\mathcal{P}_t(S)|=2^{\Oh(\sqrt{k}\log k)}$ and the algorithm computes $2^{\Oh(\sqrt{k}\log k)}\cdot n^{\Oh(1)}$ values of $\cost[t,U,P]$ for each $t\in V(T)$. Therefore,  the running time of the algorithm is $2^{\Oh(\sqrt{k}\log k)}\cdot n^{\Oh(1)}$.
\end{proof}

In this section, we discussed optimization problems but, in many cases, similar dynamic programming can be applied for counting problems.  For example, we can compute the number of (inclusion) maximal independent sets, maximal bipartite subgraphs, minimal (connected) feedback vertex sets, minimal connected vertex covers  in time  $2^{\Oh(\sqrt{k}\log k )}\cdot n^{\Oh(1)}$ on  $\chordalmke$.

\section{Beyond induced $d$-colorable subgraphs}\label{sec:lower}

In Section~\ref{sec:subexpalgo}, among other algorithms,  we gave a subexponential (in $k$) algorithm on 
$\chordalmke$ graphs that finds a maximum $d$-colorable subgraph. In particular, this also implies that for every fixed $d$, deciding whether a graph from  
$\chordalmke$ is $d$-colorable, can be done in time subexponential in $k$. In this section we show that two fundamental  problems, namely, \probCOL and 
\probCL, while still being \classFPT,  but unlikely be solvable in subexponential parameterized time.  
\subsection{Coloring  $\chordalmke$ graphs}
First, we consider \probCOL whose task is,  given a graph $G$ and a positive integer $\ell$, decide whether the chromatic number of $G$ is at most $\ell$, that is, if  $G$ is $\ell$-colorable. Note that $\ell$ here is not a fixed constant as in Section~\ref{sec:subexpalgo} and may be arbitrarily large.

Cai \cite{Cai03a} proved that \probCOL  is   \classFPT   (parameterized by $k$) on $\splitmke$ graphs. We generalize his result by showing that 
 \probCOL is \classFPT on  $\chordalmke$. Our approach is based on  the dynamic programming which is similar to the one we used in Section~\ref{sec:subexpalgo}. 
 
 We need the following property of graph colorings.
Let $G$ be a graph. It is well-known that an $\ell$-coloring of a graph $G$ can be seen as a partition $\mathcal{X}$ of $V(G)$ into at most $\ell$ independent sets formed by the vertices of the same color that are called \emph{color classes}. We call $\mathcal{X}$ an \emph{$\ell$-coloring partition} of $G$.  We also say that a partition $\mathcal{X}$ of $V(G)$ into  independent sets is  a \emph{coloring partition} of $G$. 
Let $\mathcal{X}=\{X_1,\ldots,x_\ell\}$ be a coloring partition of $G$ and 
let $U\subseteq V(G)$. We say that the partition $\mathcal{X}|_U$ of $U$ formed by nonempty sets $U\cap X_1,\ldots,U\cap X_\ell$ is a \emph{projection} of $\mathcal{X}$ on $U$;  we assume that $\mathcal{X}_\emptyset$ contains the unique element $\emptyset$. 

\begin{lemma}\label{lem:proj}
Let $G\in \chordalmke$ and let $A\subseteq \binom{V(G)}{2}$ such that $|A|\leq k$ and $G'=G+A$ is a chordal graph. Let also $C$ be a  clique of $G'$.  
Then there are at most $(2k)^{2k}$ partitions of $C$ into independent in $G$ sets 
 and these partitions can be enumerated in $2^{\Oh(k\log k)}\cdot n$ time.
\end{lemma}

\begin{proof}
Let $S$ be the set of end-vertices of the elements of $A$ in $C$. Then $C'=C\setminus S$ is a clique of $G$ such that each vertex $v$ of $C$ is adjacent to every  vertex of $C'$ distinct from $v$. Therefore,  for every coloring partition $\mathcal{X}$ of $G$, the vertices of $C'$ form single-element sets of the projection of $\mathcal{X}$ on $C$. Therefore, only the vertices of $S$ may be included in nontrivial sets of the partitions. Since $|S|\leq 2k$, there are at most $(2k)^{2k}$ partitions of $S$ into independent sets. Hence,  there are at most $(2k)^{2k}$ partitions of  $C$. To enumerate all the partitions, we do the brute force enumeration of the partitions of $S$ and verify for each partition, whether this is a partition of $S$ into independents sets. 
Since $|C|\leq n$, the enumeration can be done in $2^{\Oh(k\log k)}\cdot n$ time.   
\end{proof}

Lemma~\ref{lem:proj} implies that there are at most  $(2k)^{2k}$ projections of the coloring partitions of $G$ on $C$ and these projections can be enumerated in $2^{\Oh(k\log k)}\cdot n$ time.
This 
allows to construct a dynamic programming algorithm for \probCOL. 

\begin{theorem}\label{thm:col}
\probCOL can be solved in time $2^{\Oh(k\log k)}\cdot n^{\Oh(1)}$ on $\chordalmke$.
\end{theorem}

\begin{proof}
Since the algorithm is constructed along the same lines as the algorithms from Section~\ref{sec:subexpalgo},  we only briefly sketch the main idea.

Let $(G,\ell)$ be an instance of \probCOL.
By Lemma~\ref{lem:constr-almost}, one can construct a nice $k$-almost chordal tree decomposition  $\mathcal{T}=(T,\{X_t\}_{t\in V(T)})$ of $G$ with $\Oh(n^2)$ bags. Let $r$ be the root of $T$
and recall that $V_t$ denotes the union of the bags in the subtree of $T$ rooted in $t$. For $t\in V(T)$,
let $\mathcal{Y}_t$ be the family of all partitions of $X_t$ into independent sets. We put  $\mathcal{Y}_t=\{\emptyset\}$ if $X_t=\emptyset$. 

For every $\mathcal{X}\in \mathcal{Y}_t$, we define 
$${\sf Col}_t(\mathcal{X})=
\begin{cases}
true&\mbox{if there is an $\ell$-coloring partition }\mathcal{X}'\text{ of }G[V_t]\text{ s.t. }\mathcal{X}'|_{X_t}=\mathcal{X},\\
false&\mbox{otherwise}.
\end{cases}
$$

We compute the table of values of ${\sf Col}_t(\mathcal{X})$ bottom-up for $t\in V(V)$ starting from the leaves. Clearly, $G$ is $\ell$-colorable if and only if ${\sf Col}_r(\emptyset)=true$.

\medskip
\noindent {\bf Leaf node.} Computing the tables for leaves is trivial as the bags are empty and we assume that the empty graph is $\ell$-colorable.  

\medskip
From now on we assume  that a node $t\in V(T)$ has children and the tables are already constructed fro them. 

\medskip
\noindent {\bf Introduce node.} Let $t$ be an introduce node of $T$. Denote by $t'$ its child and assume that $X_t=X_{t'}\cup\{v\}$ for $v\notin X_{t'}$. For $\mathcal{X}\in\mathcal{Y}_t$, we define $\mathcal{X}-v$ as  the partition of $X_{t'}$ obtained from $\mathcal{X}$ either by the deletion of $\{v\}$ if $\{v\}$ is an element of $\mathcal{X}$ or by the deletion of $v$ from a nontrivial set of $\mathcal{X}$ containing $v$. For every $\mathcal{X}\in\mathcal{Y}_t$, we set
$${\sf Col}_t(\mathcal{X}):=(|\mathcal{X}|\leq \ell)\wedge {\sf Col}_{t'}(\mathcal{X}-v).
$$
 
\medskip
\noindent {\bf Forget node.} Let $t$ be a forget node of $T$. Denote by $t'$ its child and assume that $X_t=X_{t'}\setminus\{v\}$ for $v\in X_{t'}$. For $\mathcal{X}\in\mathcal{Y}_t$, we define 
$\mathcal{X}+v\subseteq \mathcal{Y}_{t\rq{}}$ to be the sets of all partitions $\mathcal{X}\rq{}$ of $X_{t\rq{}}$ into independent sets such that  $\mathcal{X}=\mathcal{X}\rq{}-v$, that is, every $\mathcal{X}\rq{}\in\mathcal{X}+v$ is either obtained by adding the single-element set $\{v\}$ or by including $v$ into one of the independent set. 
  For every $\mathcal{X}\in\mathcal{Y}_t$, we set
$${\sf Col}_t(\mathcal{X}):= \bigvee_{\mathcal{X}\rq{}\in\mathcal{X}+v}{\sf Col}_{t'}(\mathcal{X}\rq{}).
$$

\medskip\noindent {\bf{Join node.}}  Let $t$ be a join node of  $T$ with children $t_1$ and $t_2$. Then $\mathcal{Y}_t=\mathcal{Y}_{t_1}=\mathcal{Y}_{t_2}$.  
 For every $\mathcal{X}\in\mathcal{Y}_t$, 
$${\sf Col}_t(\mathcal{X}):= {\sf Col}_{t_1}(\mathcal{X})\wedge {\sf Col}_{t_2}(\mathcal{X}).$$

The correctness of the computation of the tables of values of ${\sf Col}_t(\mathcal{X})$ for $t\in\ V(T)$ is proved by standard arguments. By Lemma~\ref{lem:proj}, we have that each table has at most $(2k)^{2k}=2^{\Oh(k\log k)}$ elements. The same lemma together with the description of the computing the tables for the leave and the  introduce, forget and join nodes implies that the computation of the table for every $t\in V(T)$ can be done in time $2^{\Oh(k\log k)}n$. Therefore, the total running time is   $2^{\Oh(k\log k)}\cdot n^{\Oh(1)}$.
\end{proof}

Now we show that it is unlikely that \probCOL can be solved in subexponential in $k$ time.   For this, we show the complexity lower bound based on ETH. 
For this, we use the result of Komusiewicz and Uhlmann~\cite{KomusiewiczU12} for  the auxiliary \textsc{Triangle Cover} problem that, given a graph $G$ with $n=3p$ vertices, asks whether $V(G)$ can be covered by $p$ disjoint triangles, that is, by $p$ copies of $K_3$ that are subgraphs of $G$. 

\begin{proposition}[\cite{KomusiewiczU12}]\label{lem:triangle-cover}
Unless ETH is false, \textsc{Triangle Cover} cannot be solved in time $2^{o(n+m)}\cdot(n+m)^{\Oh(1)}$ even when the input restricted to the graphs without induced $K_4$. 
\end{proposition}

Now we rule out the existence of a subexponential algorithm for \probCOL on $\chordalmke$ parameterized by $k$. In fact, we show a stronger claim.

\begin{theorem}\label{thm:col-lower}
\probCOL cannot be solved in time $2^{o(k)}\cdot n^{\Oh(1)}$ on graphs in $\compmke$ unless ETH fails. 
\end{theorem}

\begin{proof}
A set of cliques $\{C_1,\ldots,C_\ell\}$ of a graph $G$ is a \emph{clique cover} if the cliques are pairwise disjoint and $V(G)=\cup_{i=1}^\ell C_i$.
 It is well-known that $G$ is $\ell$-colorable if and only if its complement $\overline{G}$ has a clique cover of size $\ell$. Let $G$ be a graph with $n=3\ell$ vertices that has no induced $K_4$. Observe that $G$ has a clique cover of size $\ell$ if and only if $V(G)$ can be covered by $\ell$ cliques of size 3, that is, by triangles. Consider a graph $G$ with $n=3\ell$ vertices such that $\overline{G}$ has no induced $K_4$. We obtain that $(G,\ell)$ is a yes-instance of \probCOL if and only if $\overline{G}$ is a yes-instance of \textsc{Triangle Cover}. Observe that $G\in\compmke$ if and only if $\overline{G}$ has at most $k$ edges.  Then the existence of an algorithm for \probCOL running in time $2^{o(k)}\cdot n^{\Oh(1)}$ on graphs in $\compmke$ would  imply that 
\textsc{Triangle Cover} can be solved in time $2^{o(k)}\cdot n^{\Oh(1)}$ on graphs without induced $K_4$ with at most $k$ edges. This is impossible unless ETH is false by Proposition~\ref{lem:triangle-cover}.
\end{proof}
 \subsection{Clique in $\chordalmke$ graphs} 
Now we consider the \probCL problem that asks, given a graph $G$ and a positive integer $\ell$, whether $G$ has a clique of size at least $\ell$. We show that \textsc{Clique} is \classFPT on $\chordalmke$ when parameterized by $k$ even for the weighted variant of the problem in Section~\ref{sec:kern} by demonstrating that the problem admits a Turing kernel. Here, we prove the following lower bound. 

\begin{theorem}\label{thm:clique-lower}
\probCL cannot be solved in time $2^{o(k)}\cdot n^{\Oh(1)}$ on graphs in $\compmke$ unless ETH fails. 
\end{theorem}

\begin{proof}
Observe that a graph $G$ has a clique of size at least  $\ell$ if and only if $\overline{G}$ has an independent set of size at least $\ell$.  
Recall that $G\in\compmke$ if and only if $\overline{G}$ has at most $k$ edges.  
It was shown by Impagliazzo, Paturi and Zane~\cite{ImpagliazzoPZ01} that \probIS cannot be solved in time $2^{o(n+m)}n^{\Oh(1)}$ on graphs with $n$ vertices and $m$ edges unless ETH fails. These  observations immediately imply the claim of the theorem.
\end{proof}

We proved that \probCOL and \probCL do not admit subexponential algorithms on $\compmke$, when parameterized by $k$, unless ETH fails. Note that by Observation~\ref{obs:bounds}, this means that these problems cannot be solved by subexponential algorithms on $\chordalmke$ as well unless ETH is false.

\section{Kernelization on \chordalmke}\label{sec:kern}
In this section we discuss kernelization of the problems considered in the previous section.  

Jansen and Bodlaender in \cite{JansenB13} and Bodlaender, Jansen and Kratsch in 
\cite{BodlaenderJK14} 
proved that \probWIS, \probWVC , \probWBS,  \probWOCT, \probWFVS and \textsc{Clique}
do not admit a polynomial kernel parameterized by the size of the minimum vertex cover of the  graph unless $\compass$.

We use the following observation.

\begin{observation}\label{obs:split}
If $G$ has a vertex cover of size at most $k$, then $\splitc(G)\leq\binom{k}{2}$.
\end{observation}

\begin{proof}
Let $G$ be a graph with a vertex cover $X$ of size at most $k$. Let $A=\binom{X}{2}\setminus E(G[X])$. Clearly, $G+A$ is a split graph.  Therefore, $\splitc(G)\leq \binom{k}{2}$. 
\end{proof}

Observation~\ref{obs:split} together with the results of~\cite{BodlaenderJK14,JansenB13} give the following proposition.

\begin{proposition}\label{prop:no_polykernel}
\probWIS, \probWVC , \probWBS,  \probWOCT, \probWFVS and \textsc{Clique} do not admit a polynomial in $k$ kernel on $\splitmke$ graphs unless $\compass$.
\end{proposition}

By Observation~\ref{obs:bounds}, these problems  parameterized by $k$ have no polynomial kernel on $\chordalmke$ as well unless $\compass$.

These results do not refute the existence of polynomial Turing kernels. We show that \probWCL has such a kernel. The input of \probWCL contains a graph $G$ together with a weight function $w\colon V(G)\rightarrow \mathbb{Z}^+$ and a nonnegative integer $W$, and the task is to decide whether $G$ has a clique $C$ of weight at least $W$.

Our kernelization algorithm uses the following well-known property of chordal graphs.

\begin{proposition}[\cite{Gavril74a,TarjanY84}]\label{prop:max-cl}
An $n$-vertex chordal graph has at most $n$ inclusion-maximal cliques and they can be listed in linear time.
\end{proposition}

\begin{theorem}\label{thm:int-compr}
\probWCL on $\chordalmke$ parameterized by $k$ admits a Turing kernel with at most $16k^2$ vertices with size $\Oh(k^8)$. 
\end{theorem}

\begin{proof}
Let $(G,w,W)$ be an instance of \probWCL. 

We apply Proposition~\ref{prop:fillin-appr} to approximate the fill-in of $G$. If the algorithm reports that $\fillin(G)>k$, we report that $G\notin\chordalmke$ and stop.
Assume that this is not the case. Then the algorithm returns a set $A\subseteq \binom{V(G)}{2}$ of size at most $8k^2$ such that $G'=G+A$ is a chordal graph. 
We define $X$ to be the set of vertices that are the end-vertices of the edges of $A$. Note that $|X|\leq 16k^2$. 
Then we use Proposition~\ref{prop:max-cl} to list all the inclusion maximal cliques $C_1,\ldots,C_r$ of $G'$. Let $X_i=X\cap C_i$ for $i\in\{1,\ldots,r\}$. Observe that if $C$ is a clique of $G$ of weight at least $W$, then $C\subseteq C_i$ for some $i\in\{1,\ldots,r\}$. Observe also that each $G[C_i]$ contains a clique of weight at least $W$ if and only if $G_i=G[X_i]$ contains a clique of 
weight at least $W_i=\max\{0,W-w(C_i\setminus X_i)\}$. 
Trivially, $\fillin(G_i)\leq\fillin(G)$ for $i\in\{1,\ldots,r\}$, that is, each $G_i\in\chordalmke$ if $G\in\chordalmke$.

Observe that  $(G,w,W)$ is a yes-instance of \probWCL if and only if $(G_i,w,W_i)$ is a yes-instance for at least one $i\in\{1,\ldots,r\}$.  
Then Turing kernelization algorithm   solves \probWCL  by calling the oracle for each instance $(G_i,w,W_i)$ with the parameter $k$ for $i\in \{1,\ldots,r\}$.

Since $|V(G_i)|=|X_i|\leq |X|\leq 16k^2$ for $i\in\{1,\ldots,r\}$, we solve instances with the input graphs of bounded size.
To be able to call the oracle that solves the instances of bounded size, we have to compress the weights. For this, we apply Proposition~\ref{prop:compression}. 

Let $i\in\{1,\ldots,r\}$ and consider the instance $(G_i,w,W_i)$. 
Let $v_1,\ldots,v_p$ be the vertices of $G_i$.  Using the notation  of Proposition~\ref{prop:compression}, let $h=p+1$ and $N=p+2$. Consider  vector $w=(w(v_1),\ldots,w(v_p),W_i)^\intercal\in \mathbb{Z}^h$. The algorithm of Frank and Tardos finds a vector 
$\bar{w}=(\bar{w}_1,\ldots,\bar{w}_p,\bar{W}_i)^\intercal$ with $\|\bar{w}\|_{\infty}\leq 2^{4h^3}N^{h(h+2)}$
such that  $\sign(w\cdot b)=\sign(\bar{w}\cdot b)$ for all vectors $b\in \mathbb{Z}^h$ with $\|b\|_1\leq N-1$. 
 We define $\bar{w}(v_j)=\bar{w}_j$ for $j\in\{1,\ldots,p\}$ and consider the instance $(G_i,\bar{w},\bar{W}_i)$ of \probWIS. 

Since  $\sign(w\cdot b)=\sign(\bar{w}\cdot b)$ for all vectors $b\in \mathbb{Z}^h$ with $\|b\|_1\leq N-1$, the equality holds for every $b$ whose elements are in $\{-1,0,1\}$. This implies that the weights $\bar{w}(v_j)$ are positive and $\bar{W}_i$ is nonnegative. Also we have that for every set of indices $I\subseteq\{1,\ldots,p\}$, 
$\sum_{j\in I}w(v_j)\leq W_i$ if and only if $\sum_{j\in I}\bar{w}(v_j)\leq \bar{W}_i$. This proves  that the instances $(G_i,w,W_i)$ and $(G_i,\bar{w},\bar{W}_i)$ are equivalent. 

Since $|V(G_i)|\leq 16k^2$, we obtain that the weights in the instance $(G_i,\bar{w},\bar{W}_i)$ can be encoded by stings of length $\Oh(k^6)$ for $i\in\{1,\ldots,r\}$. Hence, the size of each instance is $\Oh(k^8)$. 

Propositions~\ref{prop:fillin-appr}, \ref{prop:max-cl} and \ref{prop:compression} immediately imply that the kernelization algorithm runs in polynomial time.
\end{proof}

\section{Independent Set on $\intervalmke$}\label{sec:interval}
In this section we show that \probIS parameterized by the size of interval completion of the input graph admits a polynomial compression into the \probWIS problem.
We state \probWIS as a decision problem, whose input  contains a graph $G$ with a weight function $w\colon V(G)\rightarrow \mathbb{Z}^+$ and a nonnegative integer $W$, and the task
is to decide whether $G$ has an independent set $S$ with $w(S)\geq W$.

\noindent
More formally, we show the following theorem.

\begin{theorem}\label{thm:int-compr}
\probIS on $G\in\intervalmke$ admits a compression of size $\Oh(k^{56})$ into \probWIS. 
\end{theorem}

Since  \probWIS is in \classNP{} and, consecutively, has a polynomial reduction to \probIS  that is  \classNP-complete~\cite{GareyJ79}, the theorem immediately gives the following corollary.

\begin{corollary}\label{cor:kern-int}
\probIS on $G\in\intervalmke$ admits a polynomial kernel when parameterized by $k$. 
\end{corollary}

The remaining part of the section contains the proof of Theorem~\ref{thm:int-compr}. In Subsection~\ref{sec:int-prelim}, we introduce additional notions and state some auxiliary results. Then, in Subsection~\ref{sec:int-compr}, we give the compression itself.

\subsection{Technical lemmas}\label{sec:int-prelim}
An interval graph has been defined as an intersection graph of a family of intervals of the real line but for our compression algorithm we need the characterization of interval graphs in terms of forbidden induced subgraphs.

Three pairwise nonadjacent vertices of a graph form an \emph{asteroidal triple (AT)} if there is a path between every two of them that avoids the closed neighborhood of the third. 
For an asteroidal triple $T$ of a graph $G$, a $T$-AT-witness is an inclusion minimal  induced subgraph $F$ of $G$ such that $T$ is an asteroidal triple of $F$.  The vertices of $T$ are called \emph{terminals} of $F$. Clearly, $F$ is induced by the vertices of induced paths between every two vertices of $T$ that avoid the closed neighborhood of the third.
 This implies that the existence of an asteroidal triple $T$ can be checked in polynomial time and then the construction of a $T$-AT-witness can be done in polynomial time using the self-reducibility technique.
An \emph{asteroidal witness} is an inclusion minimal  induced subgraph that contain an asteroidal triple and we call the verices of an asteroidal triple \emph{terminals} of the witness (note that the choice of terminals may be not unique).   

It is said that a graph $G$ is \emph{AT-free} if $G$ has no asteroidal triple. 
We use the following classical result of Lekkerkerker and Boland~\cite{LekkerkerkerB62}.

\begin{proposition}[\cite{LekkerkerkerB62}]\label{prop:LB}
A graph $G$ is an interval graph if and only if $G$ is chordal and AT-free.  
\end{proposition}

The main result of Lekkerkerker and Boland~\cite{LekkerkerkerB62} is   the characterization of interval graphs by  forbidden induced subgraphs. We use this characterization in the following form tailored for our purposes.

\begin{figure}[tb]
\centering
\scalebox{0.7}{
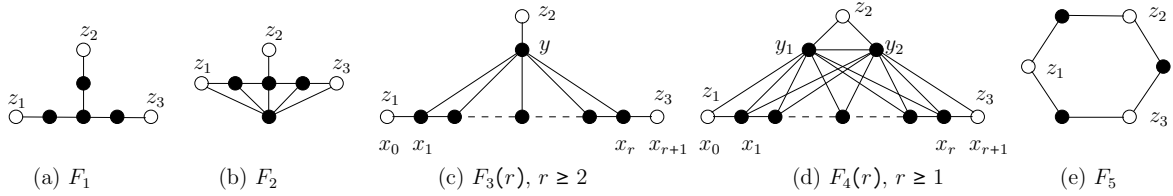}
\caption{Minimal asteroidal witnesses; asteroidal triples (terminals) are shown by white bullets.}
\label{fig:forb-int}
\end{figure}

\begin{lemma}[\cite{LekkerkerkerB62}]\label{lem:LB}
If a chordal graph $G$ contains an asteroidal triple $T$, then it contains a minimal asteroidal witness $F$ isomorphic to one of the graphs $F_1$, $F_2$, $F_3(r)$ for $r\geq 2$ or $F_4(r)$ for $r\geq 1$ that are shown in Figure~\ref{fig:forb-int} (a)--(d) and for every nonterminal vertex $v$ of $F$, it holds that $v\in V(G)\setminus T$. Moreover, such a witness $F$ can be found in polynomial time. 
\end{lemma}

The last statement of the lemma means that the vertices of an asteroidal triple of $G$ may be only terminal vertices of a minimal witness.

We say that vertex set  $X\subseteq V(G)$ is a \emph{chordal-complementing} set for $G$ if there is $A\subseteq \binom{X}{2}\setminus E(G[X])$ such that $G+A$ is chordal. Note that if $G+A$ is chordal  then the set of end-vertices of the edges of $A$ is a chordal-complementing set. Observe also that every superset of a chordal-complementing set is chordal-complementing, and for each $U\subseteq V(G)$, $X\setminus U$ is a chordal-complementing set of $G-U$.

We say that a triple of vertices $T$ of $G$ is an \emph{$X$-touching AT} if $T$ is an asteroidal triple of $G-E(G[X])$ that has a $T$-AT-witness $F$ such that either $|V(F)\cap X|\leq 1$ or $V(F)\cap X\subseteq T$. We say that $F$ is \emph{associated} with $T$.

Our compression algorithm uses the properties of chordal-complementing sets and $X$-touching ATs with associated witnesses given in the following two lemmas.

\begin{lemma}\label{lem:touching-triples} 
Let  $X\subseteq V(G)$ be a chordal-complementing set for a graph $G$ and let $T$ be an $X$-touching AT triple of $G$. 
Then $G$ has an $X$-touching AT $T'$ with an associated witness 
$F'$ isomorphic to one of the graphs $F_1$, $F_2$, $F_3(r)$ for $r\geq 2$, $F_4(r)$ for $r\geq 1$ or $F_5$ that are shown in Figure~\ref{fig:forb-int} (a)--(e). Moreover, $T'$ and 
an associated witness $F'$  can be constructed in polynomial time.
\end{lemma}

\begin{proof}
Suppose that $T=\{u_1,u_2,u_3\}$ is an  $X$-touching AT.  We find an associated witness $F$. 
Note that this can be done in polynomial time using the self-reducibility technique. 

Assume that $F$ is not a chordal graph. Then it contains an induced cycle $C$ with at least four vertices. Notice that since $X$ is a chordal-complementing set, $C$ contains at least two nonadjacent vertices from $T$. 
We consider four cases depending on the length of $C$.

Let $C$ be of length 4. By symmetry, we assume without loss of generality that $u_1,u_2\in V(C)$. Then $C=u_1v_1u_2v_2u_1$ for some $v_1,v_2\in V(F)\setminus T$. Recall that $F$ is  induced by the vertices of induced paths between every two vertices of $T$ that avoid the closed neighborhood of the third. The vertices $v_1$ and $v_2$ cannot belong to any induced $(u_1,u_3)$ or $(u_2,u_3)$-path that avoids $N_F[u_2]$ and $N_F[u_1]$.  Hence, $v_1$ and $v_2$ are vertices of an induced $(u_1,u_2)$-path that avoids $N_F[u_3]$.
Clearly, $u_1v_2u_2$ is an induced $(u_1,u_2)$-path. We obtain that $F-v_2$ is $T$-AT-witness but this contradicts the minimality of $F$. We conclude that $C$ has length at least 5.

Suppose that $C$ has length 5. Again, we can assume  without loss of generality that $u_1,u_2\in V(C)$. Then $C=u_1v_1u_2v_2v_3u_1$ for some $v_1,v_2,v_3\in V(F)\setminus T$. 
Then for every $A\subseteq \binom{X}{2}\setminus E(G[X])$, $G''=G+A$ contains a cycle of length at least four: if $u_1u_2\in A$, then $u_1u_2v_2v_3v_4v_1$ is such a cycle and if $u_1u_2\notin A$, then $C$ is a cycle of $G''$. This contradicts the condition that  $X$ is a chordal-complementing set.  Hence, $C$ has length at least 6.

Assume that $C$ has length 6. Suppose that $|V(C)\cap T|=2$. Then we can assume that $u_1,u_2\in V(C)$ and either $C=u_1v_1u_2v_2v_3v_4u_1$ or $C=u_1v_1v_2u_2v_3v_4u_1$ for some $v_1,v_2,v_3,v_4\in V(F)\setminus T$. In both cases, we obtain a contradiction with the condition that  $X$ is a chordal-complementing set in the same way as in the previous case, because 
 for every $A\subseteq \binom{X}{2}\setminus E(G[X])$, $G''=G+A$ contains a cycle of length at least four. Therefore, $T\subseteq V(C)$ and $C=u_1v_1u_2v_3u_3v_3u_1$ for some 
$v_1,v_2,v_3\in V(F)\setminus T$. We obtain that $C=F=F_5$ is a $T$-AT-witness as required by the lemma.  

Finally, let $C$ be of length at least 7. Then for two vertices of $T$, say, $u_1$ and $u_2$, $C$ contains an induced $(u_1,u_2)$-path $P$ of length at least 3. Then for every $A\subseteq \binom{X}{2}\setminus E(G[X])$, $G''=G+A$ contains a cycle of length at least four that contains $P$ as a segment. This contradicts the condition that $X$ is a chordal-complementing set.

Assume now that $F$ is a chordal graph. Then the claim of the lemma is a direct corollary of Lemma~\ref{lem:LB}.
\end{proof}

\begin{lemma}\label{lem:touching-triples-forb} 
Let  $X\subseteq V(G)$ be a chordal-complementing set for a graph $G$. Suppose that $G$ has an $X$-touching AT $T$ with an 
associated witness 
$F$ isomorphic to one of the graphs $F_1$, $F_2$, $F_3(r)$ for $r\geq 2$, $F_4(r)$ for $r\geq 1$ or $F_5$ that are shown in Figure~\ref{fig:forb-int} (a)--(e). Then for every interval complementation $H$ of $G$, $H$ has an edge $uv\notin E(G)$ such that
\begin{itemize}
\item[(i)] $u,v\in V(F)$,
\item[(ii)] either $u\notin X$ or $v\notin X$.  
\end{itemize}
\end{lemma}

\begin{proof}
Suppose that $T$ is an $X$-touching AT with an associated witness 
isomorphic to one of the graphs $F_1$, $F_2$, $F_3(r)$ for $r\geq 2$, $F_4(r)$ for $r\geq 1$ or $F_5$. Let $H$ be an interval complementation of $G$.  If $|V(F)\cap X|\leq 1$, then the claim immediately follows form Proposition~\ref{prop:LB} as $H$ is AT-free. Assume that $|V(F)\cap X|\geq 2$. Then $V(F)\cap X\subseteq T$.

\begin{figure}[tb]
\centering
\scalebox{0.7}{
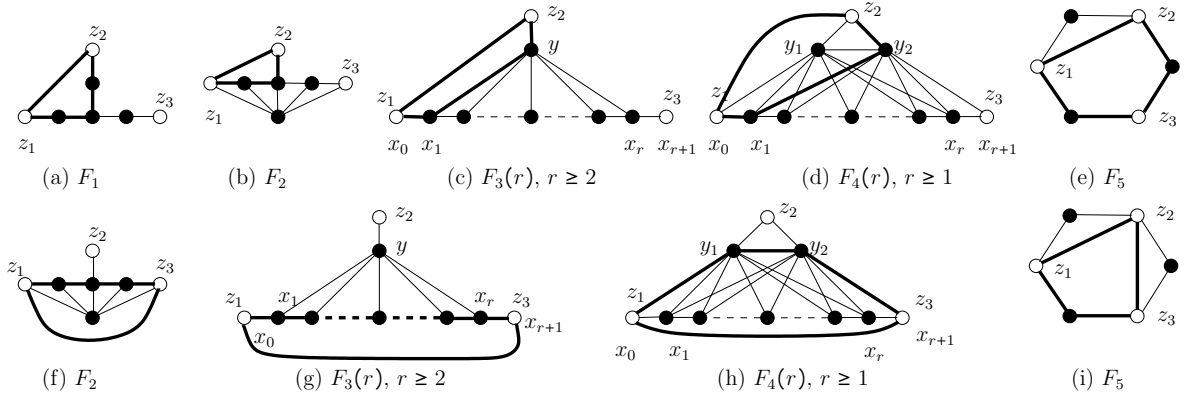}
\caption{Induced cycles in $F$ shown by thick lines.}
\label{fig:cycles}
\end{figure}

Assume that $|V(F)\cap T|=2$. If the vertices of $V(F)\cap T$ are nonadjacent in $H$, then the existence of $uv\notin E(G)$ satisfying (i) and (ii) follows from Proposition~\ref{prop:LB}. 
Let these vertices be adjacent. 
Suppose that  $z_1,z_2$ of $F$ (see Figure~\ref{fig:forb-int}) are in $T$.  Observe that $F+z_1z_2$ contains an induced cycle of length at least four as it is shown in Figure~\ref{fig:cycles}~(a)--(e). Since $H$ is chordal, we obtain that there is  $uv\in E(H)\setminus E(G)$ satisfying (i) and (ii). The case $z_2,z_3\in T$ is symmetric. Assume that $z_1,z_3\in T$. By symmetry, it is sufficient to consider the cases $F=F_2$, $F=F_3(r)$ and $F=F_4(r)$. Again, we observe that $F+z_1z_3$ contains an  induced cycle of length at least four as it is shown in Figure~\ref{fig:cycles}~(f)--(h) and the claim follows.

Let $|V(F)\cap T|=3$, that is $T=\{z_1,z_2,z_3\}$. If the vertices of $V(F)\cap T$ are pairwise nonadjacent in $H$, then the existence of $uv\in E(H)\setminus E(G)$ satisfying (i) and (ii) follows from Proposition~\ref{prop:LB}. If $H[T]$ contains an edge, then we apply the same arguments as above for the cases $F=F_2$, $F=F_3(r)$ and $F=F_4(r)$ and obtain that $F+E(H[T])$ contains an induced cycle of length at least four. This implies that  there is  $uv\in E(H)\setminus E(G)$ satisfying (i) and (ii). Let $F=F_5$. If $|E(H[T])|=1$ or $|E(H[T])|=2$, we again have that $F+E(H[T])$ contains an induced cycle of length at least four (see Figure~\ref{fig:cycles}~(e) and (i)) and the claim follows. Let $|E(H[T])|=3$. Then $F+\{z_1z_2,z_2z_3,z_1z_3\}$ coincides with $F_4(2)$. Since $H$ has no induced subgraph isomorphic to $F_4(2)$, we have that there is  $uv\in E(H)\setminus E(G)$ satisfying (i) and (ii).
\end{proof}

In our compression algorithm, we have to compute a maximum independent set for chordal graphs.
It was already observed by Gavril~\cite{Gavril72} in 1972 that this in polynomial (linear) time on chordal graphs. 

\begin{proposition}[\cite{Gavril72,RoseTL76}]\label{prop:is-chord}
\probIS can be solved in time $\Oh(n+m)$ on chordal graphs.
\end{proposition}

\subsection{Compression}\label{sec:int-compr}
In this section, we give a compression of \probIS on $G\in\intervalmke$ parameterized by $k$. Let $(G,\ell)$ be an instance of \probIS and let a nonnegative integer $k$ be the parameter.

First, we apply the algorithm of Natanzon, Shamir and Sharan~\cite{NatanzonSS98} (see Proposition~\ref{prop:fillin-appr}) to approximate the fill-in of $G$. If the algorithm reports that $\fillin(G)>k$, we immediately stop as, clearly, $G\notin\intervalmke$. Assume that this is not the case. Then the algorithm returns a set $A\subseteq \binom{V(G)}{2}$ of size at most $8k^2$ such that $G+A$ is a chordal graphs. We define $X$ to be the set of vertices that are the end-vertices of the edges of $A$. Note that $X$ is a chordal-complementing set. 
We apply a series of reduction rules for the instance of  \probIS considered together with $X$, that is, for the triple $(G,\ell,X)$.

We apply the following reduction rule to enhance $X$.

\begin{reduction}\label{red:claws}
If $G$ has an $X$-touching AT $T$ with an associated witness $F=F_1$, then set $X:=X\cup V(F)$.
\end{reduction}

We apply the rule exhaustively but at most $k+1$ times, because Lemma~\ref{lem:touching-triples-forb} guarantees that if we find an $X$-touching AT $T$ with an associated witness $F=F_1$, then
every interval complementation of $G$ contains an edge $uv$ such that $u$ and $v$ are  nonadjacent vertices of $F$ and at most one of them is in $X$. This implies that the following rule is safe.

\begin{reduction}\label{red:stop-claws}
If Reduction Rule~\ref{red:claws} have been applied $k+1$ times, then report that $G\notin\intervalmke$ 
and stop.
\end{reduction}

Assume that the algorithm did not stop. Since $|V(F_1)|=7$, we obtain that $|X|\leq 8k^2+7k$ after this step. 

In the next step, we find and delete some irrelevant vertices of $G$. For this, set $p=8k^2+7k+2$.

\begin{reduction}\label{red:irr}
If for some vertex $x\in V(G)$, the subgraph $G[N_G(x)\setminus X]$ has an independent set of size at least $p+1$, then set $G:=G-x$ and $X:=X\setminus\{x\}$.
\end{reduction}

\begin{lemma}\label{lem:safe-irr}
Reduction Rule~\ref{red:irr} is safe.
\end{lemma}

\begin{proof}
Denote by $G'$ and $X'$ the graph and the vertex set obtained from $G$ and $X$ respectively by an application of the rule for some $x\in V(G)$. Note that $X'=X$ if $x\notin X$.  
We show that $G$ is an independent set of size $\ell$ if and only $G'$ has an independent set of the same size. It is trivial that if $G'$ has an independent set of size $\ell$, then the same set is an independent set of $G$. Assume that $G$ has an independent set $I$ if size $\ell$. We prove that $G'$ has an independent set of size at least $\ell$. The claim is straightforward if $x\notin I$. Suppose that $x\in I$. Let $S$ be an independent set of size at least $p+1$ in $G[N_G(x)\setminus X]$. Clearly, $S\cap I=\emptyset$. 

If there is $S'\subset S$ of size $p-1$ that has no vertex adjacent to a vertex of $I\setminus (X\cup \{x\})$, then consider $I'=(I\setminus (X\cup\{x\}))\cup S'$. Observe that $I'$ is an independent set of $G'$. Since $|X|\leq p-2$, $|I'|\geq |I|\geq \ell$ and we have that $I'$ is a required independent set of size at least $\ell$ in $G'$. 

Assume from now that every $S'\subset S$ of size $p-1$ has a vertex with a neighbor in $I\setminus (X\cup \{x\})$. Because $|S|\geq p+1$, there are three distinct vertices $u_1,u_2,u_3\in S$ that have neighbors in $I\setminus (X\cup \{x\})$.  Denote these neighbors of $u_1$, $u_2$ and $u_3$ by $v_1$, $v_2$ and $v_3$ respectively. Note that $u_1,u_2,u_3$ are pairwise nonadjacent and they are adjacent to $x$. Notice also that $u_1,u_2,u_3,v_1,v_2,v_3\notin X$. If $v_i=v_j$ for some distinct $i,j\in \{1,2,3\}$, we have that $xu_iv_iu_jx$ is an induced cycle but this contradicts the property that $X$ is a chordal-complementing set. Hence, $v_1,v_2,v_3$ are pairwise distinct and $v_i$ is not adjacent to $u_j$ for distinct $i,j\in\{1,2,3\}$. 
Because $v_1,v_2,v_3\in I$, these vertices are pairwise nonadjacent. But then $G[\{x,u_1,u_2,u_3,v_1,v_2,v_3\}]$ is isomorphic to $F_1$ (see Figure~\ref{fig:forb-int} (a)). This contradict the assumption that Reduction Rule~\ref{red:claws} was applied exhaustively and the algorithm was not stopped by Reduction Rule~\ref{red:stop-claws}. This competes the safeness proof.
\end{proof}

We apply Reduction Rule~\ref{red:irr} exhaustively. The crucial property that we achieve by this rule is the following.

\begin{observation}\label{obs:size-F}
If $G$ has an $X$-touching AT $T$ with an associated witness 
$F=F_3(r)$ or $F=F_4(r)$ (see Figure~\ref{fig:forb-int} (c) and (d)), then $r\leq 2p-1=16k^2+14k+1$.
\end{observation} 

\begin{proof}
Assume that $G$ has an $X$-touching AT $T$ with an associated witness 
$F=F_3(r)$. Note that only the terminals of $F$ could be in $X$.  Then $\{x_1,x_3,\ldots,x_{2\lceil r/2\rceil-1}\}$ is an independent set in $G[N_G(y)\setminus X]$ (see Figure~\ref{fig:forb-int} (c)). Since Reduction Rule~\ref{red:irr} cannot be applied, $r\leq 2p-1$. If $F=F_4(r)$, the arguments are the same and the only difference is that we consider $y_1$ instead of $y$ (see Figure~\ref{fig:forb-int} (d)). 
\end{proof}

Now we proceed with enhancing $X$.

\begin{reduction}\label{red:AT}
If $G$ has an $X$-touching AT $T$ with an associated witness
$F$ isomorphic to one of the graphs $F_1$, $F_2$, $F_3(r)$ for $r\geq 2$, $F_4(r)$ for $r\geq 1$ or $F_5$ that are shown in Figure~\ref{fig:forb-int} (a)--(e), then set $X:=X\cup V(F)$.
\end{reduction}

We apply the rule exhaustively but at most $k+1$ times, because in the same way as for Reduction Rule~\ref{red:stop-claws}, we can apply 
 Lemma~\ref{lem:touching-triples-forb}. Hence, the next rule  is safe.

\begin{reduction}\label{red:stop-AT}
If Reduction Rule~\ref{red:AT} have been applied $k+1$ times, then report that $G\notin\intervalmke$ 
and stop.
\end{reduction}

Assume that the algorithm did not stop. By Observation~\ref{obs:size-F}, each $F$, whose vertices have been added to $X$ by Reduction Rule~\ref{red:AT}, has at most $2p+4=32k^2+2k+4$ vertices. Then
\begin{equation}\label{eq:size-X}
|X|\leq 8k^2+7k+k(32k^2+2k+4)=32k^3+10k^2+11k=\Oh(k^3).
\end{equation}

By applying Reduction Rule~\ref{red:AT}, we achieve an important property of $G$ and $X$. Since the rule cannot be applied any more, by Lemma~\ref{lem:touching-triples},  the following holds.

\begin{observation}\label{obs:no-AT}
The graph $G'=G-E(G[X])$ has no asteroidal triple $T\subseteq X$ with a $T$-AT-witness $F$ such that $V(F)\cap X=T$.
\end{observation}

Next, we simplify the instance $(G,\ell)$ of \probIS by the removal chordal components. It is straightforward to see that the following rule is safe.

\begin{reduction}\label{red:rem-chord}
If $G$ has a connected component $H$ that is a chordal graph, then compute the size $\alpha(H)$ of a maximum independent set of $H$ and set $G:=G-V(H)$, $X:=X\setminus V(H)$ and $\ell=\ell-\alpha(H)$. If $\ell\leq 0$, then return a trivial yes-instance of \probWIS and stop.
\end{reduction}

The rule is applied exhaustively. We assume that the algorithm did not stop. 

 For a set $Y\subseteq X$ of size at most two, let $\mathcal{C}_Y$ be the set of connected components of the graph $G-(X\cup N_G(Y)$ and define 
\begin{equation*}\label{eq:conn}
\mathcal{C}=\bigcup_{Y\subseteq X,~|Y|\leq 2}\mathcal{C}_Y.
\end{equation*}
Notice that $\mathcal{C}$ is a set of connected induced subgraphs of $G-X$ and distinct subgraphs in the  set can have common or adjacent vertices. For each $C\in\mathcal{C}$, let $I(C)$ be a maximum independent set of $C$. We use the following crucial property of $\mathcal{C}$.

\begin{lemma}\label{lem:comp-crucial}
There is a set $S\subseteq X$ and a family $\mathcal{C}^*\subseteq \mathcal{C}$ of pairwise disjoint graphs without adjacent (in $G$) vertices in distinct subgraphs 
such that 
$$I^*=S\cup\bigcup_{C\in \mathcal{C}^*}I(C)$$
is a maximum independent set of $G$.
\end{lemma}

\begin{proof}
Let $I$ be a maximum independent set of $G$. Let $S=I\cap X$ and denote by $C_1,\ldots,C_s$ the connected components of $G-(X\cup N_G(S))$. Clearly, $I\setminus S\subseteq V(C_1)\cup\ldots\cup V(C_s)$. Moreover, if $I_i$ is an arbitrary maximum independent set of $C_i$ for $i\in \{1,\ldots,s\}$, then $I'=S\cup I_1\cup\ldots\cup I_s$ is a maximum independent set of $G$. We claim that for each $i\in\{1,\ldots,s\}$, there is a set $Y\subseteq S$ of size at most two such that $C_i\in\mathcal{C}_Y$.

To obtain a contradiction, assume that there is $i\in\{1,\ldots,s\}$ such that $C_i\notin\mathcal{C}_Y$ for any $Y\subseteq S$ of size at most two.  Then $|S|\geq 3$ and there are distinct vertices $u_1,u_2,u_3\in S$ such that for every $j\in \{1,2,3\}$, there is $v_j\in N_G(u_j)\setminus X$ such that 
(i) $v_j\notin N_G(u_h)$ for $h\in\{1,2,3\}\setminus\{j\}$ and (ii) $v_j$ is adjacent to some vertex $w_j$ of $C_i$.  Consider any two distinct vertices $u_j$ and $u_h$ for $j,h\in\{1,2,3\}$. Let $t$ be the unique element of $\{1,2,3\}\setminus \{j,h\}$.
Since $C_i$ is connected, there is a $(w_j,w_h)$-path $P$ in $C_i$. Notice that $P$ avoids $N_G[u_t]$.  Let $P'=u_jv_jPv_hu_h$. We have that $P'$ is a $(u_j,u_h)$-path in $G'=G-E(G[X])$ that avoids the neighborhood of $u_t$. Since this holds for any choice of $j$ and $h$, we obtain that $T=\{u_1,u_2,u_3\}\subseteq X$ is an asteroidal triple $T$ in $G'$ with a $T$-AT-witness $F$ such that $V(F)\cap X=T$ but this contradicts Observation~\ref{obs:no-AT}. This proves that for each $i\in\{1,\ldots,s\}$, there is a set $Y\subseteq S$ of size at most two such that $C_i\in\mathcal{C}_Y$.

We obtain that $C_1,\ldots,C_s\in \mathcal{C}$ are pairwise disjoint graphs without adjacent vertices in distinct subgraphs and 
$$I^*=S\cup I(C_1)\cup\ldots\cup I(C_s)$$
is a maximum independent set of $G$.
\end{proof}

We show that $\mathcal{C}$ has size that is bounded by a polynomial of the parameter. First, we show an auxiliary claim.

\begin{lemma}\label{lem:comp-number}
Let $x\in Y\subseteq X$, where $|Y|\leq 2$, and let $\mathcal{C}'\subseteq\mathcal{C}_Y$ be the set of graphs in $\mathcal{C}_Y$ that have at least one neighbor in $N_G(x)\setminus X$. Then 
$|\mathcal{C}'|\leq 3p$ for $p=8k^2+7k+2$.
\end{lemma}

\begin{proof}
To obtain a contradiction, assume that $|\mathcal{C}'|>3p$, that is, there are at least $3p+1$ connected components of $G'=G-(N_G(Y)\cup X)$ that have neighbors in $N_G(x)\setminus X$. Let $y\in N_G(x)\setminus X$. If $y$ is adjacent to at least $p+1$ connected components of $G'$, then $N_G(y)\setminus X$ contains an independent set of size at least $p+1$ and we would be able to apply Reduction Rule~\ref{red:irr}; a contradiction. Therefore, each $y\in N_G(x)\setminus X$ has neighbors in  at most $p$ connected components of $G'$. Since $\mathcal{C}'\geq 3p+1$, we conclude that there are three distinct vertices $y_1,y_2,y_3\in N_G(x)\setminus X$ such that there are three distinct connected components $C_1,C_2,C_3\in \mathcal{C}'$ with the property that for every $i\in\{1,2,3\}$, $y_i$ has a neighbor $z_i\in V(C_i)$ and $y_i$ has no neighbor in $C_j$ for $j\in\{1,2,3\}\setminus\{i\}$. Consider $H=G[\{x,y_1,y_2,y_3,z_1,z_2,z_3\}]$. 
It is easy to see that $z_1,z_2,z_3$ is an asteroidal triple in this graph and $H$ contains an induced subgraph isomorphic to $F_1$, $F_3(2)$ or $F_3(3)$ (see Figure~\ref{fig:forb-int} (a) and (c)) depending on the adjacencies between $y_1$, $y_2$ and $y_3$. 
This means that we would be able to apply Reduction Rule~\ref{red:AT} contradicting the assumption that the rule was applied exhaustively. 
This proves that $|\mathcal{C}'|\leq 3p$.
\end{proof}

\begin{lemma}\label{lem:C-size}
$$|\mathcal{C}|=\Oh(k^{14}).$$
\end{lemma}

\begin{proof}
By the definition, $\mathcal{C}=\bigcup_{Y\subseteq X,~|Y|\leq 2}\mathcal{C}_Y$. We upper bound $|\mathcal{C}_Y|$ for $Y\subseteq X$ of size at most two.

Let $Y=\emptyset$. Observe that each connected component of $G$ contains a vertex  of $X$, because $G$ has no connected component that is a chordal graph as Reduction Rule~\ref{red:rem-chord} cannot be applied. Note that if a vertex $x$ has neighbors in at least $p+1$ connected components of $G-X$ for $p=8k^2+7k+2$, then $N_G(x)\setminus X$ contains an independent set of size at least $p+1$ and we would be able to apply Reduction Rule~\ref{red:irr}. Therefore, each vertex of $X$ has neighbors in at most $p$ components of $G-X$. Hence, 
$|\mathcal{C}_\emptyset|\leq p|X|=\Oh(k^5)$ by the definition of $p$ and (\ref{eq:size-X}).

Suppose that $Y=\{x\}$ for $x\in X$. Denote by $\mathcal{C}^{(1)}$ the set of graphs in  $\mathcal{C}_Y$ that have vertices adjacent to $N_G(x)\setminus X$ and let $\mathcal{C}^{(2)}$ be the set of graphs in  $\mathcal{C}_Y$ that have no vertex adjacent to $N_G(x)\setminus X$. We have that $\mathcal{C}_Y=\mathcal{C}^{(1)}\cup \mathcal{C}^{(2)}$. By exactly the same arguments as for $Y=\emptyset$, we obtain that $|\mathcal{C}^{(2)}|\leq p(|X|-1)$. By Lemma~\ref{lem:comp-number}, $|\mathcal{C}^{(1)}|\leq 3p$.
We obtain that
$$ |\mathcal{C}_Y|=|\mathcal{C}^{(1)}\cup \mathcal{C}^{(2)}|\leq |\mathcal{C}^{(1)}|+|\mathcal{C}^{(2)}|\leq 3p+p(|X-1|)$$
and, therefore, $|\mathcal{C}_Y|=\Oh(k^5)$ by the definition of $p$ and (\ref{eq:size-X}). 

Suppose now that $Y=\{x_1,x_2\}$ for distinct $x_1,x_2\in X$. Denote by $\mathcal{C}^{(i)}$ the set of graphs in  $\mathcal{C}_Y$ that have vertices adjacent to some vertices of $N_G(x_i)\setminus X$ for $i\in\{1,2\}$ and 
and let $\mathcal{C}^{(3)}$ be the set of graphs in  $\mathcal{C}_Y$ that have no vertex adjacent to $N_G(Y)\setminus X$. We have that 
$\mathcal{C}_Y=\mathcal{C}^{(1)}\cup \mathcal{C}^{(2)}\cup \mathcal{C}^{(3)}$. 
By exactly the same arguments as for $Y=\emptyset$, we obtain that $|\mathcal{C}^{(3)}|\leq p(|X|-2)$. By Lemma~\ref{lem:comp-number}, 
$|\mathcal{C}^{(1)}|\leq 3p$ for $i\in\{1,2\}$. Then
$$ |\mathcal{C}_Y|=|\mathcal{C}^{(1)}\cup \mathcal{C}^{(2)}\cup \mathcal{C}^{(3)}|\leq |\mathcal{C}^{(1)}|+|\mathcal{C}^{(2)}|+|\mathcal{C}^{(3)}|\leq 6p+p(|X-2|)$$
and, therefore, $|\mathcal{C}_Y|=\Oh(k^5)$ by the definition of $p$ and (\ref{eq:size-X}).

Since there are $|X|$ single-element subsets $Y\subseteq X$ and $\binom{|X|}{2}$ two-element subsets $Y\subseteq X$, we have that $|\mathcal{C}|=\Oh(k^{14})$ by (\ref{eq:size-X}).
\end{proof}

\paragraph{Construction of the instance of \probWIS.} At the next step of our compression algorithm we construct the instance of  \probWIS as follows.
\begin{itemize}
\item Construct the graph $G^*$ with the vertex set $X\cup \mathcal{C}$ by making every two distinct vertices $u$ and $v$ either adjacent or nonadjacent  by the following rule:
\begin{itemize}
\item if $u,v\in X$, then $u$ and $v$ are adjacent in $G^*$ if and only if they are adjacent in $G$,
\item if $u\in X$ and $v\in \mathcal{C}$, then $u$ and $v$ are adjacent if and only if $u$ is adjacent to a vertex of the subgraph $v$ in $G$,
\item if $u,v\in\mathcal{C}$, then $u$ and $v$ are adjacent if and only if the subgraph $u$ and $v$ of $G$ have either common vertices or adjacent vertices in $G$.
\end{itemize}
\item For $v\in V(G^*)$, set the weight $w(v)=1$ if $v\in X$ and set $w(v)$ be the size of a maximum independent set of the subgraph $v$ of $G$.
\end{itemize}

\begin{lemma}\label{lem:equiv}
The instance $(G,\ell)$ is a yes-instance of \probIS if and only if $(G^*,w,\ell)$ is a yes-instance of \probWIS.
\end{lemma}

\begin{proof}
Let $(G,\ell)$ be a yes-instance of \probIS.  By Lemma~\ref{lem:comp-crucial}, there is a set $S\subseteq X$ and a family $\mathcal{C}^*\subseteq \mathcal{C}$ of pairwise disjoint graphs without adjacent vertices in distinct subgraphs 
such that 
$$I=S\cup\bigcup_{C\in \mathcal{C}^*}I(C)$$
is a maximum independent set of $G$, where $I(C)$ is a maximum independent set of $C$. By the definition of $G^*$, we have that $S\cup \mathcal{C}^*$ is an independent set of $G^*$ of weight
$|S|+\sum_{C\in \mathcal{C}^*}|I(C)|=|I|$. Hence, $(G^*,w,\ell)$ is a yes-instance of \probWIS.

Assume now that $(G^*,w,\ell)$ is a yes-instance of \probWIS. Consider an independent set $I^*$  of weight at least $\ell$ in $G^*$.  Let $S=I^*\cap X$ and $\mathcal{C}^*=I^*\cap \mathcal{C}$. 
By the definition of $G^*$, we have that $S\subseteq X$ is an independent set of $G$, and the graphs of $\mathcal{C^*}$ are disjoint induced subgraphs of $G$ that have no vertices adjacent to $S$ and there are no two adjacent vertites that are in distinct graphs of  $\mathcal{C}^*$. Every graph $C\in \mathcal{C}^*$ has an independent set $I(C)$ of size $w(C)$ by the definition of the weights.
This means that $I=S\cup\bigcup_{C\in \mathcal{C}^*}I(C)$ is an independent set in $G$ of size at least $\ell$. Therefore, $(G,\ell)$ is a yes-instance of \probIS. 
\end{proof}

By Lemma~\ref{lem:C-size}, $G^*$ has $\Oh(k^{14})$ vertices, that is, the size of $G^*$ is bounded by a polynomial of the parameter. To complete the construction of the compressed instance, it remains to reduce the weights of vertices. We do it by making use of Proposition~\ref{prop:compression}. Let $v_1,\ldots,v_s$ be the vertices of $G^*$. Following the notation of Proposition~\ref{prop:compression}, let $h=s+1$ and $N=s+2$. Consider the vector $w=(w(v_1),\ldots,w(v_s),\ell)^\intercal\in \mathbb{Z}^h$. The algorithm of Frank and Tardos finds a vector 
$\bar{w}=(w_1,\ldots,w_s,W)$ with $\|\bar{w}\|_{\infty}\leq 2^{4h^3}N^{h(h+2)}$
such that  $\sign(w\cdot b)=\sign(\bar{w}\cdot b)$ for all vectors $b\in \mathbb{Z}^h$ with $\|b\|_1\leq N-1$. 
 We define $w^*(v_i)=w_i$ for $i\in\{1,\ldots,s\}$ and consider the instance $(G^*,w^*,W)$ of \probWIS. This completes the construction of the compression. The properties of the obtained instance of \probWIS are summarized in the following lemma.

\begin{lemma}\label{lem:size-compr}
The instance $(G^*,w,\ell)$ is a yes-instance of \probWIS if and only if $(G^*,w^*,W)$ is a yes-instance. 
The size of $(G^*,w^*,W)$ is $\Oh(k^{56})$.
\end{lemma}

\begin{proof}
Notice that, in particular, the equality $\sign(w\cdot b)=\sign(\bar{w}\cdot b)$ holds for  all vectors $b\in \mathbb{Z}^h$ such that each element of $b$ is $-1$, $0$ or $1$. This implies that the elements of $\bar{w}$ are positive and for every $J\subseteq\{1,\ldots,s\}$, $\sum_{i\in J}w(v_i)\geq \ell$ if and only if $\sum_{i\in J}w_i\geq W$. Clearly, for every set of vertices $S\subseteq V(G^*)$, $\sum_{v\in S}w(v)\geq \ell$ if and only if $\sum_{v\in S}w^*(v)\geq W$. This means that $(G^*,w,\ell)$ is a yes-instance of \probWIS if and only if $(G^*,w^*,W)$ is a yes-instance. 

Since $\|\bar{w}\|_{\infty}\leq 2^{4h^3}N^{h(h+2)}$, we have that $w(v)\leq 2^{4(s+1)^3}(s+2)^{(s+1)(s+3)}$ for every $v\in V(G^*)$ and the same upper bound holds for $W$. This implies that the weights of the vertices and $W$ can be encoded by a string of length $\Oh(k^{42})$, because by Lemma~\ref{lem:C-size}, $|V(G^*)|=\Oh(k^{14})$. Because $G^*$ has $\Oh(k^{14})$ vertices and 
$\Oh(k^{28})$ edges, the size of $(G^*,w^*,W)$ is $\Oh(k^{56})$.
\end{proof}

\paragraph{Running time evaluation.}
Finally, we have to show that the compression algorithm is polynomial. The construction of the initial set $X$ can be done in polynomial time by Proposition~\ref{prop:fillin-appr}. Then we apply Reduction Rules~\ref{red:claws}--\ref{red:rem-chord}. It is straightforward to see that Reduction Rule~\ref{red:claws} can be applied in polynomial time as we are looking in it for an induced subgraph of bounded size. Reduction Rules~\ref{red:stop-claws} and \ref{red:stop-AT} are trivial. Reduction Rules~\ref{red:irr} and  \ref{red:rem-chord} are polynomial by Proposition~\ref{prop:is-chord}. Reduction Rule~\ref{red:AT} is polynomial by Lemma~\ref{lem:touching-triples}. Since Reduction Rules~\ref{red:claws} and \ref{red:AT} are applied at most $k+1$ times and Reduction Rules~\ref{red:irr} and \ref{red:rem-chord} are applied at most $|V(G)|$ times, we conclude all the rules can be applied in polynomial time. In the next step, we construct the instance $(G^*,w,\ell)$ and the step is polynomial due to Proposition~\ref{prop:is-chord}. Finally, we reduce the weight and this can be done in polynomial time by Proposition~\ref{prop:compression}.

\medskip
This completes the proof of Theorem~\ref{thm:int-compr}.

\section{Independent Set on $\splitmke$}\label{sec:clique}
In this section, we show that \probIS admits a polynomial kernel when parameterized by the split completion size. 

 It is known to be  NP-hard to compute $\splitc(G)$~\cite{NatanzonSS01} but, interestingly, if we allow not only edge additions but also edge deletions, then the problem becomes polynomial time solvable as it was proved by Hammer and Simeone in~\cite{HammerS81}.
 
 \begin{proposition}[\cite{HammerS81}]\label{prop:split-edit}
 There is an algorithm that, given a graph $G$, in polynomial time finds a set $A\subseteq\binom{V(G)}{2}$ of minimum size such that $G\bigtriangleup A$ is a split graph.
 \end{proposition}
 
 \begin{theorem}\label{thm:split-kern}
 \probIS on $\splitmke$ admits a polynomial kernel with at most  $2k^2(k+2)$ vertices when parameterized by $k$. 
  \end{theorem}
  
 \begin{proof}
 Let $(G,\ell)$ be an instance of \probIS and let a nonnegative integer $k$ be the parameter.

We use Proposition~\ref{prop:split-edit} and find a set $A\subseteq\binom{V(G)}{2}$ of minimum size such that $G'=G\bigtriangleup A$ is a split graph. If $|A|>k$, we conclude that 
$G\notin\splitmke$
and stop. Assume that is not the case.  Let $D=A\cap E(G)$ and $C=A\setminus E(G)$, that is, $D$ is the set of deleted edges and $C$ is the set of added edges. 
We find a partition of $V(G')$ into a clique $K$ and an independent set $I$. Note that by the minimality of $A$, the edges of $D$ have their end-vertices in $I$ and the edges of $C$ have their end-vertices in $K$.
Let $X$ be the set of end-vertices of $C$ and set $Y=K\setminus X$. Note that $|X|\leq 2k$ and every vertex of $X$ is adjacent to each vertex of $Y$.  

We apply a series of reduction rules for $(G,\ell)$ together with the sets $I,X,Y,D$.   

\begin{reduction}\label{red:incr-X}
If $D\neq \emptyset$, then for $uv\in D$ do the following.
\begin{itemize}
\item If  each of $u$ and $v$ is not adjacent to at least $k+2$ vertices of $Y$ in $G$, then report that $G\notin\splitmke$
and stop.
\item Otherwise, pick a vertex of $\{u,v\}$ with the minimum number of nonneighbors  in $Y$, say $u$, and 
\begin{itemize}
\item set $D:=D-uv$,
\item set $I:=I\setminus \{u\}$,
\item set $X:=X\cup\{u\}\cup(Y\setminus N_G(u))$,
\item set $Y:=Y\cap N_G(u)$.
\end{itemize}
\end{itemize}
\end{reduction}  
 
We apply the rule exhaustively, until $D$ becomes empty.  

\begin{claim}\label{cl:stop}
If the algorithm stops while executing Reduction Rule~\ref{red:incr-X}, then $\splitc(G)>k$. 
\end{claim} 

\begin{subproof}[Proof of Claim~\ref{cl:stop}]
The algorithm stops if there is $uv\in D$ such that $|Y\setminus N_G(u)|\geq k+2$ and $|Y\setminus N_G(v)|\geq k+2$. Let $R\subseteq \binom{V(G)}{2}\setminus E(G)$ be a set of pairs of vertices of minimum size such that $\hat{G}=G+R$ is a split graph. Let $(S,Z)$ be a partition of $V(\hat{G})$ into an independent set $S$ and a clique $Z$.  Notice that either $u\in Z$ or $v\in Z$. By symmetry, assume without loss of generality that $u\in Z$. Observe that $Y$ is a clique of $G$.  Hence, $|Y\setminus Z|\leq 1$.  This implies that $R$ contains at least $k+1$ edges incident to $u$ in $\hat{G}$ whose other end-vertices are in $Y$. Therefore, $\splitc(G)=|R|\geq k+1$.  
\end{subproof}

Claim~~\ref{cl:stop} guarantees that if we stop by Reduction Rule~\ref{cl:stop}, then $G\notin\splitmke$.  Assume that the algorithm did not stop. Then we obtain that the constructed sets $X$, $Y$ and $I$ have the properties that are summarized in the following claim. 

\begin{claim}\label{cl:prop}
The sets $X,Y,I$ form a partition of the vertices of $G$ such that 
\begin{itemize}
\item[(i)] $I$ is an independent set in $G$,
\item[(ii)] $Y$ is a clique in $G$,
\item[(iii)] for every $v\in X$, $Y\subseteq N_G(v)$,
\item[(iv)] $|X|\leq (k+2)k$,
\item[(v)] for every independent set $S\subseteq X$, $|S|\leq 2k$. 
\end{itemize}
\end{claim} 
  
\begin{subproof}[Proof of Claim~\ref{cl:prop}]
It is straightforward that $(X,Y,I)$ is a partition of $V(G)$.

To see (i), it is sufficient to observe that only the edges of $D$ had both their end-vertices in $I$ in the initial  $I$ and we exclude at least one  end-vertex of every edge of $D$ from $I$ by Reduction Rule~\ref{red:incr-X}.

The property (ii) is trivial as $Y$ was a clique before we started to apply Reduction Rule~\ref{red:incr-X} and we only delete vertices from $Y$ by the rule. 

For (iii), observe that if $v$ is a vertex of the initial set $X$, then $Y\subseteq N_G(v)$ by the definition of $X$ and $Y$. Then, if we add a vertex $u\in I$ to $Y$ by Reduction Rule~\ref{red:incr-X}, then we delete the vertices of $Y\setminus N_G(u)$ from $Y$. Note that these vertices are included in $X$ and, since $Y$ is a clique, they are adjacent to all remaining vertices of $Y$.
Hence, $Y\subseteq N_G(v)$ for every $v\in X$.
  
To show (iv), notice that initially $|X|\leq 2|C|$. Then, whenever we apply Reduction Rule~\ref{red:incr-X}, we add to $X$ at most $k+2$ vertices. The rule is applied at most $|D|$ times.
We obtain that $|X|\leq 2|C|+(k+2)|D|\leq (k+2)|A|\leq (k+2)k$.

Finally, to prove (v), observe that initially $|X|\leq 2|C|$ and, therefore, every independent set with its vertices in the initial set $X$ has size at most $2|C|$. By each application of  Reduction Rule~\ref{red:incr-X}, we put a vertex $u\in I$  in $X$ and add a clique $Y\setminus N_G(u)$. This means that we can increase the maximum size of independent subset of $X$ by at most two. The rule is applied at most $|D|$ times. We conclude that the maximum size of independent subset of $X$ is at most $2|C|+2|D|=2|A|\leq 2k$.
\end{subproof}

These properties allow us to apply the next two rules.

\begin{reduction}\label{red:del-Y}
Set $Y:=Y\setminus N_G(I)$.
\end{reduction}  

\begin{reduction}\label{red:red-Y}
If $|Y|\geq 2$, then delete $|Y|-1$ arbitrary vertices of $Y$ and set $I:=I\cup Y$.
\end{reduction}  

\begin{claim}\label{cl:del-Y}
Reduction Rules~\ref{red:del-Y} and \ref{red:red-Y} are safe and the set $I$ constructed by Reduction Rule~\ref{red:red-Y} is independent. 
\end{claim} 
  
\begin{subproof}[Proof of Claim~\ref{cl:del-Y}]
Let $G'$ be the graph obtained from $G$ by the application of Reduction Rules~\ref{red:del-Y} and \ref{red:red-Y} and denote by $I'$ the set obtained from $I$. For the safeness proof, it is sufficient to show that if $G$ has an independent set $S$ with $|S|\geq \ell$, then $G'$ has an independent set of size at least $\ell$.  If $S\cap Y=\emptyset$, $S\subseteq V(G')$, that is, $S$ is an independent set of $G'$. Suppose that there is $v\in S\cap Y$. Since $Y$ is a clique, $v$ is the unique vertex of $S$ in $Y$. Since $Y\subseteq N_G(u)$ for every $u\in X$, $X\cap S=\emptyset$. This means that $S\setminus \{v\}\subseteq I$.

Suppose that there is $u\in I$ such that $v\in N_G(u)$. Consider $S'=(S\setminus \{v\})\cup\{u\}$. We obtain that $S'$ is an independent set, since $S'\subseteq I$. Clearly, $|S'|=|S|\geq \ell$.

Assume now that $v\in Y\setminus N_G(I)$. Note that the vertices of $Y\setminus N_G(I)$ are true twins in $G$, that is, for every $x,y\in Y\setminus N_G(I)$, $N_G[x]=N_G[y]$. Hence, we can assume without loss of generality that $u$ was not deleted by Reduction Rule~\ref{cl:del-Y} and $S\subseteq V(G')$.

To see that $I'$ constructed by Reduction Rule~\ref{red:red-Y} is independent, it is sufficient to observe that we include in $I$ a unique vertex of $Y$ that is not adjacent to other vertices of $I$.
\end{subproof}

Observe that after applying Reduction Rules~\ref{red:del-Y} and \ref{red:red-Y}, we have that $(X,I)$ is a partition of $V(G)$, where $I$ is an independent set.

\begin{reduction}\label{red:red-I}
If there is a vertex $u\in X$ such that $|N_G(u)|\geq 2k$, then set $X:=X\setminus \{u\}$.
\end{reduction} 

\begin{claim}\label{cl:red-I}
Reduction Rule~\ref{red:red-I} is safe.
\end{claim}

\begin{subproof}[Proof of Claim~\ref{cl:red-I}]
Denote by $G'$ the graph obtained from $G$ by the application of Reduction Rule~\ref{red:red-I} for a vertex $u\in X$. It is sufficient to show that if $G$ has an independent set $S$ with $|S|\geq \ell$, then $G'$ has an independent set of size at least $\ell$.  If $u\notin  S$, then $S$ is an independent set of $G'$. Assume that $u\in S$. Consider the set $I$. Observe that
$|S\cap I|\leq |I|-2k$. By Claim~\ref{cl:prop} (v), $|S\cap X|\leq 2k$. We obtain that $|I|\geq |S|$. Since $I$ is an independent set in $G'$, the claim follows.
\end{subproof}

We apply the rule exhaustively. We obtain $G$ with the property that  $|N_G(u)\cap I|\leq 2k-1$ for every $u\in X$.

Finally, we exhaustively apply the following reduction rule. 

\begin{reduction}\label{red:red-isol}
If there is an isolated vertex $u$, then set $G:=G-u$, $X:=X\setminus\{u\}$, $I:=I\setminus\{u\}$ and $\ell:=\ell-1$.
\end{reduction} 

It is straightforward to see that the rule is safe. 

Now we can show that the obtained graph $G$ has bounded size.

\begin{claim}\label{cl:size-G}
The graph $G$ has at most $2k^2(k+2)$ vertices.
\end{claim}

\begin{subproof}[Proof of Claim~\ref{cl:size-G}]
After the exhaustive application of Reduction Rule~\ref{red:red-isol}, $G$ has no isolated vertices. Recall that by Claim~\ref{cl:prop} (iv), $|X|\leq (k+2)k$. Since each vertex of $X$ is adjacent to at most $2k-1$ vertices in $I$, $|I|\leq (2k-1)(k+2)k$. Then $|V(G)|\leq |X|+|I|\leq (k+2)k+(2k-1)(k+2)k=2k^2(k+2)$.  
\end{subproof}

This completes the construction of the kernel. 

The initial construction of $X$, $Y$, $D$ and $I$ is polynomial by Proposition~\ref{prop:split-edit}. It  is straightforward to see that Reduction Rules~\ref{red:incr-X}--\ref{red:red-isol} can be applied in polynomial time. Hence, our kernelization algorithm is polynomial.
\end{proof}

\section{Conclusion}\label{sec:concl}
In this  paper, we initiated the study of parameterized subexponential and kernelization algorithms on $\chordalmke$ graphs. The existence of such algorithms makes this graph class  a very interesting object for studies. For other structural parameters, like treewidth or vertex cover, we have quite good understanding about the complexity of various optimization problems derived from general meta-theorems like Courcelle's or Pilipczuk's theorems~\cite{Courcelle90,DBLP:journals/corr/abs-1104-3057} and advanced algorithmic techniques
\cite{cut-and-count,rank-treewidth2,FominLPS16}. We believe that further exploration of the 
 complexity landscape of fill-in parameterization is an interesting research direction. 
 If an optimization problem is \classNP-complete on chordal graphs, like \probDS, then on 
 $\chordalmke$ this problem is in \classParaNP.
  On the other hand, even if a problem is solvable in polynomial time on chordal graphs, in theory, there is nothing preventing it from being \classParaNP on  $\chordalmke$. Is there a natural graph problem with this property? 
  For many problems that are solvable in polynomial time on chordal graphs, we 
 also established \classFPT algorithms on  $\chordalmke$ class. This does not exclude a possibility that there are  problems that are not \classFPT parameterized by $k$ but solvable in polynomial time for every fixed $k$.
 We do not know any such problem (in other words, the problem in class \classXP)
 yet. 
 It will be interesting to see, if there is any natural graph problem of such complexity.  
  In addition, we proved that there are problems  that are  \classFPT   on $\chordalmke$ when parameterized by $k$ and which cannot be solved in subexponential time unless ETH fails. We believe it would be exciting to 
   obtain a logical characterization of problems that can be solved in subexponential time on  $\chordalmke$ when parameterized by $k$, similar to the classical Courcelle's theorem~\cite{Courcelle90}.

\medskip
Some concrete open problems. Observe that for our subexponential   dynamic programming algorithms, we only need a $k$-almost chordal tree decomposition of the input graph, that is, a decomposition where each bag can be made a clique by adding at most $k$ edges. 
(Recall Definition~\ref{def:almost}.)
The maximum of numbers $\cliqc(G[X_t])\leq k$ can be significantly smaller than the minimum fill-in of a graph.    For graphs in $\chordalmke$, we can find fill-in in a subexponential in $k$ time by the algorithm of Fomin and Villanger~\cite{FominV13}.  However, we do not know  if it is  \classFPT in $k$ to decide, whether a graph admits a $k$-almost chordal tree decomposition. And  if yes,  can it be done in subexponential time?

The only reason why many of the algorithms introduced in our work run in time 
$2^{\Oh(\sqrt{k}\log k)}\cdot n^{\Oh(1)}$ and not 
$2^{\Oh(\sqrt{k})}\cdot n^{\Oh(1)}$ is because we do not know how to find a fill-in of size $k$  in time 
$2^{\Oh(\sqrt{k})}\cdot n^{\Oh(1)}$.  
The best known lower bound  rules out algorithms of time $2^{o(\sqrt{k})}\cdot n^{\Oh(1)}$  \cite{CaoS20} and better algorithms for fill-in would imply better algorithms for various optimization problems on  $\chordalmke$ graphs. Moreover, to get rid of the logarithm in the exponent, we do not need an exact algorithm. 
By the results of Natanzon, Shamir and Sharan~\cite{NatanzonSS98} (Proposition~\ref{prop:fillin-appr}),  $\fillin(G)$ can be approximated in polynomial time within a polyopt factor $8\cdot \fillin(G)$. Deciding whether $\fillin(G)\leq k$ can be done in time $2^{\Oh(\sqrt{k}\log k)}\cdot n^{\Oh(1)}$ by the results of 
Fomin and Villanger~\cite{FominV13} (Proposition~\ref{prop:FV12}). Is there an \classFPT constant-factor approximation algorithm with running time $2^{\Oh(\sqrt{k})}\cdot n^{\Oh(1)}$? The existence of such an algorithm would improve running times of the algorithms for many problems. For example, pipelined with our results, it would imply that  \probWIS is solvable in  $2^{\Oh(\sqrt{k})}\cdot n^{\Oh(1)}$ time on \chordalmke.

 Finally, we proved that \probIS  on $\intervalmke$ and $\splitmke$ admit  polynomial  kernels when parameterized by $k$. We leave open the question whether or not this problem has a polynomial 
(Turing) kernel on $\chordalmke$.

\medskip\noindent\textbf{Acknowledgement.} We thank Torstein Str\o mme, Daniel Lokshtanov, and Pranabendu Misra for fruitful discussions on the topic of this paper.  We also grateful to Saket Saurabh for helpful suggestions that allowed us to improve our results.

\bibliographystyle{siam}
\bibliography{Chordlike}

\end{document}